\documentclass[12pt]{article}

\usepackage{amssymb}
\usepackage{amsmath}
\usepackage{amsfonts}
\usepackage{amsthm}

\theoremstyle{plain}
\newtheorem{theorem}{Theorem}[section]
\newtheorem{corollary}[theorem]{Corollary}
\newtheorem{proposition}[theorem]{Proposition}
\newtheorem{lemma}[theorem]{Lemma}

\theoremstyle{definition}
\newtheorem{definition}[theorem]{Definition}
\newtheorem{Note}[theorem]{Note}

\begin{document}

\title{Superposition rules for higher-order systems and their applications}

\author{\small J F Cari\~nena$^{1}$, J Grabowski$^{2}$ and J de Lucas$^{1,3}$\footnote{e-mail:jfc@unizar.es$^{1}$,j.grabowski@impan.pl$^{2}$,delucas@impan.pl$^{3}$}
}

\maketitle

\centerline{\small $^{1}$Departamento de F\'isica Te\'orica and IUMA, Universidad de Zaragoza,}
\centerline{\small c. Pedro Cerbuna 12, 50.009, Zaragoza, Spain.}
\centerline{\small $^{2}$Faculty of Mathematics and Natural Sciences, Cardinal Stefan Wyszy\'nski University,}
\centerline{\small W\'oycickiego 1/3, 01-938 Warszawa, Poland.}
\centerline{\small $^{3}$Institute of Mathematics, Polish Academy of Sciences,}
\centerline{\small ul. \'Sniadeckich 8, P.O. Box 21, 00-956, Warszawa, Poland.}
\begin{abstract}
{\it Superposition rules} form a class of functions that describe general solutions of systems of first-order
ordinary differential equations in terms of generic families of particular solutions and certain constants.
In this work we extend this notion and other related ones to systems of higher-order differential equations and analyse their properties. Several results concerning the existence of various types of
superposition rules for higher-order systems are proved and illustrated with examples extracted from the physics and
mathematics literature. In particular, two new superposition rules for second- and third-order Kummer--Schwarz
equations are derived.
\end{abstract}

\indent{\bf AMS}: 34A26 (Primary), 22E60, 34A34 (Secondary)\\


\section{Introduction}
\noindent

The study of superposition rules can be traced back to the end of the 19th century, when Lie, Vessiot, and
Guldberg \cite{LS}-\cite{Gu93} characterized and analysed the properties of systems of first-order
differential equations admitting this property, the so-called {\it Lie systems} \cite{CarRamGra}-\cite{CGM07}. Although the linear superposition
rule for  homogeneous linear systems of first-order differential equations admits a natural analogue for
homogeneous linear systems of higher-order differential equations (HODEs), the generalisation of their nonlinear counterpart is not so evident and
it has hardly been investigated so far \cite{Ve95,CL10SecOrd2}.

Recently, the necessity of a theory of (linear and nonlinear) superposition rules for systems of HODEs became
even more evident, as this concept repeatedly came up in the study of certain systems of second-order
differential equations with multiple applications in physics and mathematics \cite{CL10SecOrd2}-\cite{SIGMA}.

In an attempt to fill in this gap of the mathematics literature, the present work aims to formalize the
superposition rule notion for systems of HODEs and to analyse its properties. Since superposition rules for
systems of second-order differential equations (SODEs) represent one of the most relevant  types of
superposition rules appearing in the literature, special attention is paid to this case.

A notion of superposition rule for systems of SODEs was introduced in \cite{CL10SecOrd2}. Nevertheless, that
work was more focused on the practical use of the concept than on studying its properties. That is why we
start here by motivating this definition in detail and analysing some of its properties.

The fundamental problem on the analysis of superposition rules for systems of HODEs is to find coordinate-free geometric
conditions ensuring their existence. This problem, solved by the {\it Lie--Scheffers theorem} for systems of
first-order differential equations, is here explicitly solved for systems of SODEs. Our new result provides
not only a new insight into the study of superposition rules for SODEs, but also shows the existence of new
and more powerful types of superposition rules for such equations. These new notions can be regarded as
generalisations of other concepts already defined for systems of first-order differential equations (see
\cite{CGM07}). In addition, most of our achievements can be directly generalized to all systems of HODEs and
they are also employed to review previous notions dedicated to the study of such systems, e.g. SODE Lie
systems.

Apart from their mathematical interest, our results are also relevant so as to study all physical systems and problems, like nonquadratic Hamiltonians or Berry phases (see \cite{AL08} and references therein), related to differential equations admitting a superposition rule, such as second-order Riccati equations \cite{CL10SecOrd2} or Milne--Pinney equations \cite{CL08MP}.

To highlight the interest of our methods, they are illustrated by the analysis of examples extracted from the physics and mathematics
literature. Special attention is paid to second- and third-order Kummer--Schwarz equations, whose mathematical interest is
due, for instance, to their appearance in {\it Kummer's problem},  the study of Schwarzian derivatives, and other related topics \cite{Be07,Ta89}. Furthermore, Kummer--Schwarz equations occur in the analysis of non-stationary two body problems \cite{Be80,Be89} and, via their relation to Riccati and Milne--Pinney equations \cite{Co94,GGG11}, they can be employed to study several problems appearing in  cosmology, quantum mechanics, and other branches of physics \cite{AL08,Be80,Co94,GGG11}. We here derive
superposition rules for the analysis of such equations that provides us with several advantages with respect to previous methods of studying these, and other related, equations \cite{Co94,GGG11,Lie77}. As a byproduct, we find a new property of
Kummer--Schwarz equations: their dynamics is determined by a curve in a Lie algebra of vector fields isomorphic
to $\mathfrak{sl}(2,\mathbb{R})$.

The content of the paper is structured as follows. In Section 2 we describe some notions and results of the
theory of Lie systems to be used throughout the paper. Section 3 concerns the motivation and analysis of the
definition of a superposition rule for SODEs as well as several particular types of it  found in the literature.
In Section 4 we provide a characterization of systems of SODEs admitting certain types of superposition rules
and we describe a new kind of superposition rules for SODEs. In addition, several properties of superposition
rules for SODEs are analysed. The relation of our new results and the so--called SODE Lie systems is
studied in Section 5. The results of the previous sections lead to the definition and analysis, in Section 6, of a
general notion of a superposition rule for systems of first- and higher-order differential equations.
Subsequently, we illustrate in Sections 7 and 8  some of the theoretical results derived throughout our work
by the investigation of several remarkable HODEs. Finally, Section 9 summarizes
our achievements and details some work to be accomplished in the future.

\section{Fundamentals on Lie systems}\label{FLS}
\setcounter{equation}{0}
We hereafter assume all geometrical objects and mappings, like vector fields or superposition rules, to be real, smooth, and globally
defined. In this way, we highlight the key points of our presentation by omitting the analysis of certain
minor technical problems. For additional information, we refer to \cite{CGM07,CL10SecOrd2}.

\begin{definition} A {\it superposition rule} for a system of first-order ordinary differential equations
\begin{equation}\label{LieSystem}
\frac{dx^i}{dt}=X^i(t,x),\qquad i=1,\ldots,n,
\end{equation}
is a map $\Phi:\mathbb{R}^{mn} \times \mathbb{R}^n \rightarrow\mathbb{R}^n$ of the form
\begin{equation}\label{super}
x=\Phi(x_{(1)},\ldots,x_{(m)};k_1,\ldots,k_n),
\end{equation}
allowing us to write the general solution of system (\ref{LieSystem}) as
\begin{equation}\label{superposition}
x(t)=\Phi(x_{(1)}(t),\ldots,x_{(m)}(t);k_1,\ldots,k_n),
\end{equation}
with $x_{(1)}(t),\ldots,x_{(m)}(t)$ being a `generic' family of particular solutions and $k_1,\ldots,k_n$ being the constants related to the initial conditions of each particular solution.
\end{definition}

\begin{Note} We shall not define rigorously what `generic' means in the above definition, as it is not essential to our
purposes and depends on the particular case.  It shall be sufficient to bear in mind that, in the case of
linear superposition rules for  homogeneous linear systems of first-order differential equations, `generic'
means that the elements of the chosen finite family of particular solutions must be linearly independent.
\end{Note}

The uppermost achievement of the theory of Lie systems was obtained by Lie \cite{LS}, who succeeded in
characterizing systems of first-order differential equations that admit a superposition rule.

\begin{theorem}{\bf (The Lie--Scheffers theorem)} A system (\ref{LieSystem}) admits a superposition rule (\ref{super})
if and only if its right-hand side can be written as
\begin{equation}\label{LieDecom}
\frac{dx^i}{dt}=Z_1(t)\xi^i_1(x)+\ldots+Z_r(t)\xi^i_r(x),\qquad i=1,\ldots,n,
\end{equation}
so that the vector fields
\begin{equation}\label{VG}
X_\alpha(x)=\sum_{i=1}^n\xi_\alpha^i(x)\frac{\partial}{\partial x^i},\qquad \alpha=1,\ldots,r,
\end{equation}
with $r\leq m\cdot n$, span an $r$-dimensional real Lie algebra.
\end{theorem}

The following definition and lemma, whose proof is a straightforward consequence of the Jacobi identity,
notably simplify several statements and proofs of various results concerning the theory of Lie systems.
\begin{definition}\label{LieSpan} Given a (finite or infinite) family $\mathcal{A}$ of vector fields on $\mathbb{R}^n$,
we denote by ${\rm Lie}(\mathcal{A})$ the smallest Lie algebra $V$ of vector fields on $\mathbb{R}^n$ containing $\mathcal{A}$.
\end{definition}

\begin{lemma}\label{LieFam} Given a family of vector fields $\mathcal{A}$, the linear space ${\rm Lie}(\mathcal{A})$
is spanned by the vector fields of
$$\mathcal{A},\,\,[\mathcal{A},\mathcal{A}],\,\,[\mathcal{A},[\mathcal{A},\mathcal{A}]],\,\,
[\mathcal{A},[\mathcal{A},[\mathcal{A},\mathcal{A}]]],\ldots$$
where $[\mathcal{A},\mathcal{B}]$, with $\mathcal{B}=\mathcal{A},[\mathcal{A},\mathcal{A}],\dots$, denotes the set of Lie brackets between the elements of the families
$\mathcal{A}$ and $\mathcal{B}$ of vector fields.
\end{lemma}

Recall that if $\tau:{\rm T}\mathbb{R}^n\rightarrow\mathbb{R}^n$ denotes the tangent bundle projection and
$\pi_2$ stands for the projection $\pi_2:(t,x)\in\mathbb{R}\times\mathbb{R}^n\mapsto x\in\mathbb{R}^n$, a {\it
time-dependent vector field} $X$ on $\mathbb{R}^n$ is a map $X:(t,x)\in\mathbb{R}\times\mathbb{R}^n\mapsto
X(t,x)\in \mathbb{\rm T}\mathbb{R}^n$ such that $\tau\circ X=\pi_2$. Observe that every time-dependent vector
field $X$ on $\mathbb{R}^n$ can be regarded as a family $\{X_t\}_{t\in\mathbb{R}}$ of vector fields on
$\mathbb{R}^n$, where $X_t:x\in\mathbb{R}^n\mapsto X_t(x)=X(t,x)\in{\rm T}_x\mathbb{R}^n$.

Similarly to standard vector fields, time-dependent vector fields also admit integral curves \cite{Car96,FM}.
We hereafter call an {\it integral curve} of $X$ passing through $(t_0,x_0)\in \mathbb{R} \times\mathbb{R}^n$ any
integral curve $\gamma^{x_0}_{t_0}:s\in\mathbb{R}\mapsto (t(s),\bar \gamma(s))\in \mathbb{R}\times
\mathbb{R}^n$ of the one-dimensional distribution on $\mathbb{R}\times \mathbb{R}^n$ spanned by the {\it
suspension} of $X$, i.e. the vector field $\partial/\partial t+X(t,x)$  \cite{FM}, satisfying  that $(t_0,x_0)\in {\rm
Im}\,\gamma_{t_0}^{x_0}$.

From a modern geometric point of view, every system of first-order differential equations of the form
(\ref{LieSystem}) is described by the unique time-dependent vector field on $\mathbb{R}^n$, namely,
$X(t,x)=\sum_{i=1}^nX^i(t,x)\partial/\partial x^i$, whose integral curves are (up to an appropriate
reparametrisation) of the form $(t,x(t))$, with $x(t)$ being a solution of system (\ref{LieSystem}). For
simplicity, we  use the symbol $X$ to refer to both, a time-dependent vector field and the system of
differential equations describing its integral curves.

In such geometric terms, the Lie--Scheffers theorem states that a system $X$ admits a superposition rule if
and only if there exists a finite-dimensional Lie algebra of vector fields $V$, the so-called {\it
Vessiot--Guldberg Lie algebra}, such that $\{X_t\}_{t\in\mathbb{R}}\subset V$. In consequence, $X$ is a Lie
system if and only if the Lie algebra ${\rm Lie}(\{X_t\}_{t\in\mathbb{R}})$ is finite-dimensional.

The geometrical interpretation of superposition rules as well as one of the techniques for their determination
is based on the notion of {\it diagonal prolongation} \cite{CGM07}.

\begin{definition} Given a time-dependent vector field $X(t,x)=\sum_{i=1}^nX^i(t,x)\partial/\partial x^i$ on $\mathbb{R}^n$,
the time-dependent vector field $\widetilde X$ on $\mathbb{R}^{n(m+1)}$ of the form
$$
\widetilde X=\sum_{a=0}^m\sum_{i=1}^nX^i(t,x_{(a)})\frac{\partial}{\partial x^i_{(a)}},
$$
is called the {\it diagonal prolongation} to $\mathbb{R}^{n(m+1)}$ of $X$.
\end{definition}

A method for determining superposition rules is briefly described as follows (see \cite{CGM07,CL10SecOrd2} for
details and examples).
\begin{enumerate}
\item Take a basis $X_1,\ldots,X_r$ of a finite-dimensional Lie algebra (\ref{VG}) associated with the Lie system under study.
\item Choose the smallest positive integer $m$, so that the diagonal prolongations of the elements of the previous basis to $(\mathbb{R}^{n})^m$
are linearly independent at a generic point.
\item Take global coordinates $x^1,\ldots,x^n$ on $\mathbb{R}^n$. By defining this coordinate
system on each copy of $\mathbb{R}^n$ within $(\mathbb{R}^n)^{m+1}$, we get a coordinate system
$\{x^i_{(a)}\mid i=1,\ldots,n,\,\,a=0,\ldots,m\}$ on $(\mathbb{R}^{n})^{m+1}$.
Obtain $n$ functionally independent first-integrals $F_1,\ldots, F_n$ common to all diagonal prolongations
$\widetilde X_1,\ldots,\widetilde X_r$ of $X_1,\ldots,X_r$ to $(\mathbb{R}^{n})^{m+1}$ such that $\partial(F_1,\ldots,F_n)/\partial(x_{(0)}^1,\ldots,x_{(0)}^n)\neq 0$. This can be performed,
for instance, by means of the well-known {\it method of characteristics}.
\item Assume the above first-integrals to take certain real constant values, i.e. $F_i=k_i$ for $i=1,\ldots,n$.
By means of these equations, calculate the expressions of the variables $x_{(0)}^1\ldots,x_{(0)}^n$ in
terms of $x^1_{(a)},\ldots,x_{(a)}^n$, with $a=1,\ldots,m,$  and $k_1,\ldots,k_n$.
\item The obtained expressions give rise to a superposition rule in terms of any generic family of $m$ particular solutions
and the constants $k_1,\ldots, k_n$.
\end{enumerate}

Given two vector fields $X$ and $Y$, we have that $\widetilde{[X,Y]}=[\widetilde X,\widetilde Y]$, i.e. the
Lie bracket of two diagonal prolongations is a diagonal prolongation. Another, much less evident, property of
diagonal prolongations is described in the following lemma, whose proof can be found in \cite[Lemma 1]{CGM07}.

\begin{lemma}\label{FundLem}
Consider a family of vector fields $X_1,\ldots, X_r$ on $\mathbb{R}^n$ whose diagonal prolongations to
$\mathbb{R}^{nm}$ are linearly independent at a generic point. Then, given their diagonal prolongations
$\widetilde X_1,\ldots,\widetilde X_r$ to $\mathbb{R}^{n(m+1)}$, a vector field
$X=\sum_{\alpha=1}^rb_\alpha\widetilde X_\alpha$, with $b_\alpha\in C^{\infty}(\mathbb{R}^{n(m+1)})$, is again
a diagonal prolongation if and only if the functions $b_\alpha$ are constant.
\end{lemma}

It is worth noting that one can relate superposition rules to zero-curvature connections on a bundle ${\rm
pr}:(x_{(0)},\ldots,x_{(m)})\in\mathbb{R}^{n(m+1)}\mapsto(x_{(1)},\ldots,x_{(m)})\in\mathbb{R}^{nm}$ as
follows (cf. \cite{CGM07}).

\begin{proposition}\label{FolSup} Each superposition rule (\ref{superposition}) for a system $X$ is equivalent to
a local $n$-codimensional foliation on $\mathbb{R}^{n(m+1)}$ whose leaves project, by ${\rm pr}$, diffeomorphically
onto $\mathbb{R}^{nm}$ and such that the vector fields $\{\widetilde X_t\}_{t\in\mathbb{R}}$ are tangent to its leaves.
\end{proposition}

The above result can be used to easily prove the following new result which is used posteriorly in order to
analyse the existence of a particular class of superposition rules for systems of SODEs.

\begin{proposition}\label{CSR} A family of Lie systems admits a common superposition rule if and only if they admit
a common Vessiot--Guldberg Lie algebra.
\end{proposition}

\section{On the general definition of a superposition rule for SODEs}\label{Justify}
\setcounter{equation}{0}
To motivate the general definition of a superposition rule for SODEs, let us start by analysing a particular property
of standard superposition rules. It is well known that every homogeneous linear system on $\mathbb{R}^n$ of
the form
\begin{equation}\label{linear}
\frac{dx^i}{dt}=\sum_{j=1}^n A^i\,_j(t)x^j,\qquad i=1,\ldots,n,
\end{equation}
where the $A^i\,_j$ are real $t$-dependent functions, admits its general solution $x(t)$ to be written as
\begin{equation}\label{linearSup}
x(t)=k_1x_{(1)}(t)+\ldots+k_nx_{(n)}(t),
\end{equation}
with $x_{(1)}(t),\ldots,x_{(n)}(t)$ being a family of linearly independent particular solutions of
(\ref{linear}) and $k_1,\ldots,k_n$ a set of real constants. In other words, the system (\ref{linear}) admits a
{\it linear superposition rule}. This leads to the existence of nonlinear systems admitting general {\it
superposition rules} \cite{Win83}. Indeed, every diffeomorphism $\phi:\mathbb{R}^n\ni x\mapsto
z\in\mathbb{R}^n$ transforms system (\ref{linear}) into
\begin{equation}\label{nolinear}
\frac{dz^i}{dt}=F^i(t,z),\qquad i=1,\ldots,n,
\end{equation}
where the functions $F^i:\mathbb{R}^{n+1}\rightarrow\mathbb{R}$ are generally nonlinear in the variables
$z^1,\ldots,z^n$, and, what is more important, whose general solution $z(t)$ can be
expressed (maybe nonlinearly)  as
$$
z(t)=\phi(k_1 \phi^{-1}(z_{(1)}(t))+\ldots+k_n\phi^{-1}(z_{(n)}(t))),
$$
in terms of certain families of particular solutions $z_{(1)}(t),\ldots,z_{(n)}(t)$ of (\ref{nolinear}) and
the constants $k_1,\ldots,k_n$. That is, since linearity depends on coordinate systems and the existence of
superposition rules does not (recall the Lie--Scheffers theorem), the mere existence of linear superposition
rules for homogeneous linear systems of first-order differential equations leads to the existence of nonlinear
systems admitting superposition rules. In addition, it is worth noting that not every system admitting a
nonlinear superposition rule is of this form. For instance, Riccati equations admit a superposition rule, but
they cannot always be transformed diffeomorphically  into linear homogeneous systems \cite{RDM}.

The aforementioned properties have an analogue for systems of SODEs. In fact, it can easily be proved that every
homogeneous linear system of SODEs
\begin{equation}\label{SecondLinear}
\frac{d^2x^i}{dt^2}=\sum_{j=1}^n \left(A^i\,_j(t)\frac {d x^j}{dt} +B^i\,_j(t)x^j\right),\qquad i=1,\ldots,n,
\end{equation}
with $A^i\,_j$ and $B^i\,_j$ being any set of $2n^2$ time-dependent functions, admits its general
solution to be written as
\begin{equation}\label{SupSODElinear}
x(t)=k_1x_{(1)}(t)+\ldots+k_{2n}x_{(2n)}(t),
\end{equation}
in terms of some arbitrary constants $k_1,\ldots,k_{2n}$ and a set of solutions
$\{x_{(a)}(t)\,|\,a=1,\ldots,2n\}$ such that the vectors $(x_{(a)}(t),dx_{(a)}(t)/dt)\in{\rm T}\mathbb{R}^n$
are linearly independent at every $t\in\mathbb{R}$. Now, a change of variables $z=\phi(x)$ transforms the
above system into a (generally nonlinear) new one
\begin{equation}\label{SecondNolinear}
\frac{d^2z^i}{dt^2}=H^i\left(t,z,\frac{dz}{dt}\right),\qquad i=1,\ldots,n,
\end{equation}
for certain functions $H^i:{\rm T}\mathbb{R}^{n}\times\mathbb{R}\rightarrow \mathbb{R}$, and, moreover, such a
change enables us, in view of (\ref{SupSODElinear}), to write its general solution $z(t)$ in the form
\begin{equation}\label{nolinearSupSeg}
z(t)=\phi(k_1\phi^{-1}(z_{(1)}(t))+\ldots+k_{2n}\phi^{-1}(z_{(2n)}(t))),
\end{equation}
in terms of a generic family of particular solutions $z_{(1)}(t),\ldots,z_{(2n)}(t)$ for
(\ref{SecondNolinear}) and constants $k_1,\ldots,k_{2n}$. Consequently, linear superposition rules
for systems (\ref{SecondLinear}) give rise to the existence of `superposition rule-like' expressions for
systems of SODEs. Expressions of this type frequently appear in the literature, e.g. in the study of linear
inhomogeneous systems of SODEs. This suggests us the following definition that was proposed and briefly
analysed in \cite{CL10SecOrd2} and that includes the previous expressions as particular cases.

\begin{definition}
A {\it base-superposition rule} for a system
\begin{equation}\label{SODE}
\frac{d^2x^i}{dt^2}=F^i\left(t,x,\frac{dx}{dt}\right), \qquad i=1,\ldots,n,
\end{equation} is a map
$\Upsilon:(\mathbb{R}^{n})^m\times\mathbb{R}^{2n}\rightarrow\mathbb{R}^n$ allowing us to write its general
solution $x(t)$ as
\begin{equation}\label{FreeVel}
x(t)=\Upsilon(x_{(1)}(t),\ldots,x_{(m)}(t);k_1,\ldots,k_{2n}),
\end{equation}
where $x_{(1)}(t),\ldots,x_{(m)}(t)$ is a generic family of particular solutions and $k_1,\ldots,k_{2n}$ are constants.
\end{definition}

The above concept does not cover many other expressions found in the literature for describing systems of
SODEs \cite{CL10SecOrd2,CLNE07,CL08MP,SIGMA}. For instance, consider a Milne--Pinney equation
\begin{equation}\label{MilnePinney}
\frac{d^2x}{dt^2}=-\omega^2(t)x+\frac{c}{x^3},
\end{equation}
with $x>0$ and $\omega(t)$ being any time-dependent real function  \cite{Mil30}-\cite{Re99}. This equation is relevant due to its applications in quantum mechanics, cosmology, Bose--Einstein condensates, and other physical topics \cite{CL08MP,SIGMA,AL08}. Recently, it was proved
(see \cite{CL08MP}) that its general solution can be written as
\begin{equation}\label{MilnSup}
 x(t)=\left\{k_1x_{(1)}^2(t)+\!k_2x_{(2)}^2(t)\!\pm\! 2\left[\lambda_{12}[I_3x_{(1)}^2(t)x_{(2)}^2(t)-\!c(x_{(1)}^4(t)\!+\!x_{(2)}^4(t))]\right]^{1/2}\right\}^{1/2}\!\!,
\end{equation}
by means of a generic pair $x_{(1)}(t),x_{(2)}(t)$ of particular solutions, the function
$\lambda_{12}=\lambda_{12}(k_1,k_2,c,I_3)$, the constant of motion
$$I_3=\left(\frac{dx_{(1)}}{dt}(t)x_{(2)}(t)-\frac{dx_{(2)}}{dt}(t)x_{(1)}(t)\right)^2+
c\left[\left(\frac{x_{(1)}(t)}{x_{(2)}(t)}\right)^2+\left(\frac{x_{(2)}(t)}{x_{(1)}(t)}\right)^2\right],
$$
and two constants $k_1, k_2$ related to initial conditions. Observe that expression (\ref{MilnSup}) cannot be
described by means of any base-superposition rule notion. Indeed, while $k_1, k_2$ take different values to
describe the different particular solutions of (\ref{MilnSup}), the constant $I_3$, whose value is fixed by
the chosen particular solutions and their time-derivatives, does not appear in base-superposition rules. The
same will happen for other new relevant expressions to be presented in this work. This motivates us to
generalize the base-superposition rule as follows.

\begin{definition}
A {\it quasi-base superposition rule} for a system of SODEs in $\mathbb{R}^n$ of the form (\ref{SODE}) is a function $G:\mathbb{R}^{mn}\times\mathbb{R}^{q}\times\mathbb{R}^{2n}\rightarrow\mathbb{R}^n$ allowing us to
cast its general solution $x(t)$ in the form
\begin{equation}\label{express1}
x(t)=G(x_{(1)}(t),\ldots,x_{(m)}(t),I_1,\ldots,I_{q};k_1,\ldots,k_{2n}),
\end{equation}
in terms of any generic family $x_{(1)}(t),\ldots,x_{(m)}(t)$ of particular solutions of (\ref{SODE}), a set
of time-independent constants of motion $I_1,\ldots,I_{q}$, whose values are determined by the choice of the
previous family and their derivatives with respect to the time, and a set of constants $k_1,\ldots,k_{2n}$.
\end{definition}

Although almost every example of `superposition-rule like' expression for SODEs is a particular instance of
a quasi-base superposition rule, this notion still fails to cover several expressions found in the literature.
That is the case of the very recently discovered expression for second-order Riccati equations, presented in
\cite{CL10SecOrd2}, which describes the general solution of such equations in terms of a generic family of
particular solutions, their derivatives, and several constants. This motivates us to generalize the
the concept of a quasi-base superposition rule as follows.

\begin{definition}\label{Def1} A system of second-order ordinary differential equations on $\mathbb{R}^n$ given by (\ref{SODE}) admits a {\it superposition rule} if there exists a map $\Upsilon:({\rm
T}\mathbb{R}^{n})^{m}\times\mathbb{R}^{2n}\rightarrow \mathbb{R}^n$ of the form
\begin{equation}\label{SupSODE1}
x=\Upsilon(x_{(1)},v_{(1)},\ldots,x_{(m)},v_{(m)};k_1,\ldots,k_{2n}),
\end{equation}
with $(x_{(a)},v_{(a)})\in {\rm T}_{x_{(a)}}\mathbb{R}^n$ for $a=1,\ldots,m$, such that the general solution
$x(t)$ of (\ref{SODE}) can be written as
\begin{equation}\label{SupSODE}
x(t)=\Upsilon\left(x_{(1)}(t),\frac{dx_{(1)}}{dt}(t),\ldots,x_{(m)}(t),\frac{dx_{(m)}}{dt}(t);k_1,\ldots,k_{2n}\right),
\end{equation}
with $x_{(1)}(t),\ldots,x_{(m)}(t)$ being any generic family of $m$ particular solutions of the system, and
$k_1,\ldots,k_{2n}$ being a set of constants related to the initial conditions of each particular solution.
\end{definition}

As every constant of motion involved in a quasi-base superposition rule can be considered as a function on
${\rm T}\mathbb{R}^{mn}$, quasi-base superposition rules can be easily regarded as a particular type of
superposition rules for SODEs. Base-superposition rules can also be regarded as
superposition rules that depend only on the base variables of ${\rm T}\mathbb{R}^{nm}$ that justifies
their name.

Let us now turn to analysing several properties of superposition rules for systems of SODEs. Similarly to
superposition rules for systems of first-order differential equations, expression (\ref{SupSODE}) cannot be
applied for each family of $m$ particular solutions. Recall that even in the simple case of a homogeneous linear
system of SODEs, expression (\ref{linearSup}) just remains valid for certain families of particular solutions.
Consequently, in order to establish when a system (\ref{SODE}) admits a superposition rule, it is essential to
establish what `generic' means in this new context. From now on, we say that (\ref{SupSODE}) is satisfied by a
generic set of $m$ particular solutions if there exists an open and dense subset $U$ of $({\rm
T}\mathbb{R}^{n})^m$ such that expression (\ref{SupSODE}) is valid for every family
$x_{(1)}(t),\ldots,x_{(m)}(t)$ that satisfies
$$
\left(x_{(1)}(0),\frac{dx_{(1)}}{dt}(0),\ldots,x_{(m)}(0),\frac{dx_{(m)}}{dt}(0)\right)\in U\subset ({\rm
T}\mathbb{R}^n)^m.
$$
Every family of particular solutions satisfying the above condition is called a {\it fundamental system} of
particular solutions of system (\ref{SODE}).

Superposition rules for systems of SODEs possess properties different from those of superposition rules for systems of
first-order ones. Let us illustrate this fact by means of a particular remarkable difference. Consider again
(\ref{LieSystem}) as a Lie system admitting superposition rule (\ref{superposition}). A time-reparametrisation
$\tau=\tau(t)$, with the inverse $t=t(\tau)$, transforms this system into
\begin{equation}\label{transf}
\frac{dx^i}{d\tau}=\frac{dt}{d\tau}F^i(t(\tau),x),\qquad i=1,\ldots,n.
\end{equation}
As we assume (\ref{LieSystem}) to be a Lie system, formula (\ref{LieDecom}) applies and the right-hand term
of the above expression can be brought into the form
\begin{equation}\label{De2}
\frac{dx^i}{d\tau}=\frac{dt}{d\tau}\left(Z_1(t(\tau))\xi^i_1(x)+\ldots+Z_r(t(\tau))\xi^i_r(x)\right),\qquad
i=1,\ldots,n.
\end{equation}
Consequently, in view of the Lie--Scheffers theorem, the system  (\ref{transf}) becomes a Lie system. Moreover, as the
general solution $x(t)$ of  (\ref{LieSystem}) and the general solution $x(\tau)$ of (\ref{transf}) satisfy
$x(t(\tau))=x(\tau)$,  then the superposition rule (\ref{superposition}) for (\ref{LieSystem}) allows one to
write
$$
x(\tau)=\Phi(x_{(1)}(\tau),\ldots,x_{(m)}(\tau);k_1,\ldots,k_n),
$$
in terms of a generic family  $x_{(1)}(\tau),\ldots,x_{(m)}(\tau)$  of particular solutions of system (\ref{transf})
and $k_1,\ldots,k_n$. In summary, Lie's characterization of systems of first-order ordinary differential
equations admitting a superposition rule is invariant under time-reparametrisations and Lie systems related in
this way share a common superposition rule. Indeed, note that this follows trivially from the form of
(\ref{De2}) and Proposition \ref{CSR}.

The above property is no longer valid for superposition rules of systems of SODEs. Given a system of SODEs
admitting a superposition rule, the systems obtained from it by time-reparametrisations do not necessarily
possess the same superposition rule. For instance, consider a system of SODEs (\ref{SODE}) admitting a
superposition rule (\ref{SupSODE}). A time-reparametrisation $\tau=\tau(t)$, with the inverse $t=t(\tau)$,
transforms (\ref{SODE}) into
\begin{equation}\label{transfSecOrd}
\frac{d^2x^i}{d\tau^2}=\frac{d^2t}{d\tau^2}\frac{dx^i}{d\tau}\frac{d\tau}{dt}+
\left(\frac{dt}{d\tau}\right)^2F^i\left(t(\tau),x,\frac{dx}{d\tau}\frac{d\tau}{dt}\right),\quad
i=1,\ldots,n,
\end{equation}
whose general solution $x(\tau)$ satisfies $x(\tau)=x(t(\tau))$, where $x(t)$ is the general solution of
(\ref{SODE}). Hence, from the superposition rule (\ref{SupSODE}), we get that $x(\tau)$ can be expressed as
$$
x(\tau)=\Upsilon\left(x_{(1)}(\tau),\frac{d\tau}{dt}(t(\tau))\frac{dx_{(1)}}{d\tau}(\tau),
\ldots,x_{(m)}(\tau),\frac{d\tau}{dt}(t(\tau))\frac{dx_{(m)}}{d\tau}(\tau);k_1,\ldots,k_{2n}\right).
$$
The above expression is not necessarily a superposition rule, as it may admit an explicit dependence on the
new time variable $\tau$. A simple example illustrating this fact can be found in Section \ref{Examples}.

Obviously, we could have also required the superposition rule concept for systems of SODEs to be invariant
under time-reparametrisations, but this would exclude several important examples like second-order Riccati or
Milne--Pinney equations \cite{CL10SecOrd2, SIGMA}.

\section{On the existence of superposition rules for SODEs}\label{TheorySODE}
\setcounter{equation}{0}
The following theorem characterizes systems of SODEs admitting a superposition rule. We hereafter use canonical
global coordinates $(x^1,\ldots,x^n,v^1,\ldots,v^n)$ on ${\rm T}\mathbb{R}^n$. By defining this coordinate
system on each copy of ${\rm T}\mathbb{R}^n$ within $({\rm T}\mathbb{R}^n)^m$, we obtain a coordinate system
$\{x^i_{(a)},v^i_{(a)}\,|\,i=1,\ldots,n,\,\,a=1,\ldots,m\}$ on $({\rm T}\mathbb{R}^n)^m$.

\begin{theorem}\label{MT}A mapping $\Upsilon:({\rm T}\mathbb{R}^n)^m\times\mathbb{R}^{2n}\rightarrow\mathbb{R}^n$
is a superposition rule for a system of SODEs (\ref{SODE}) if and only if
\begin{enumerate}
 \item the functions $u_k:({\rm T}\mathbb{R}^n)^m\ni p\mapsto u_k(p)=\Upsilon(p;k)\in\mathbb{R}^n$, with $k\in\mathbb{R}^{2n}$,
 are common solutions for the $t$-parametrized family of systems of PDEs on $({\rm T}\mathbb{R}^n)^m$ given by
\begin{equation}\label{Char}
 \qquad (X^{(m)}_D)_t^2u^i_k+(X^{(m)}_L)_tu^i_k=F^i(t,u_k,(X^{(m)}_D)_tu_k),\qquad i=1,\ldots,n,
\end{equation}
where $u_k=(u_k^1,\ldots,u_k^n)\in\mathbb{R}^n$ and $X^{(m)}_D, X^{(m)}_L$ are the diagonal prolongations to
$({\rm T}\mathbb{R}^n)^m$ of the time-dependent vector fields
\begin{equation}\label{DefVec}
 \qquad  X_D=\sum_{i=1}^n\left(v^i\frac{\partial}{\partial {x^i}}+F^{i}(t,x,v)\frac{\partial}{\partial
{v^{i}}}\right),\,\,\,\, X_L=\sum_{i=1}^n\partial_tF^i(t,x,v)\frac{\partial}{\partial v^i}\,
\end{equation}
and
\item the map $\varphi:({\rm T}\mathbb{R}^n)^m\times\mathbb{R}^{2n}\rightarrow{\rm T}\mathbb{R}^n$ of the form
\begin{equation}\label{Fun}
\varphi(p;k)=(u_k(p),[(X_D^{(m)})_0u_k](p))\in {\rm T}_{u_k(p)}\mathbb{R}^n
\end{equation}
gives rise to a family of bijections $\varphi_p:k\in\mathbb{R}^{2n}\mapsto \varphi(p;k)\in{\rm
T}\mathbb{R}^{n}$, with $p$ being a generic point of $({\rm T}\mathbb{R}^n)^m$.
\end{enumerate}
\end{theorem}
\begin{proof} Assume that the SODE Lie system (\ref{SODE}) has a superposition rule (\ref{superposition}).
One can define the function $u_k:p\in({\rm T}\mathbb{R}^n)^m\mapsto \Upsilon(p;k)\in \mathbb{R}^n$,
for each $k\in\mathbb{R}^{2n}$,  which leads, for every fundamental system of solutions $x_{(1)}(t),\ldots,x_{(m)}(t)$
of (\ref{SODE}), to a new particular solution of this system,
\begin{equation}\label{NewSol}
\bar x(t)=u_k\left(x_{(1)}(t),\frac{dx_{(1)}}{dt}(t),\ldots,x_{(m)}(t),\frac{dx_{(m)}}{dt}(t)\right).
\end{equation}
On the other hand,
\begin{equation}\label{eq0}
\frac{d\bar{x}^i}{dt}(t)=\sum_{a=1}^m\sum_{j=1}^n\left(v^j_{(a)}\frac{\partial
u_k^i}{\partial{x^j_{(a)}}}+F^j_{(a)}\frac{\partial u_k^i}{\partial {v^j_{(a)}}}\right)(p(t)),\qquad
i=1,\ldots,n,
\end{equation}
where, for shortening the notation, we have denoted $F_{(a)}^j=F^j(t,x_{(a)},v_{(a)})$ and
\begin{equation}\label{Setting}
p(t)=\left(x_{(1)}(t),\frac{dx_{(1)}}{dt}(t),\ldots,x_{(m)}(t),\frac{dx_{(m)}}{dt}(t)\right).
\end{equation}
From the expression of $X_D$ given in (\ref{DefVec}), it follows
\begin{equation}\label{eq1}
\frac{d\bar{x}^i}{dt}(t)=[(X^{(m)}_D)_tu^i_k](p(t)),\qquad i=1,\ldots,n.
\end{equation}
By differentiating expression (\ref{eq0}) with respect to the time, we obtain
\begin{eqnarray*}
\!\!\!\! \!\frac{d^2\bar{x}^i}{dt^2}(t)=\left[\sum_{j,l=1}^n\sum_{a,b=1}^m\left(v^j_{(a)}v^l_{(b)}\frac{\partial^2
u^i_k}{\partial {x^j_{(a)}\partial x^l_{(b)}}}+2v^j_{(a)}F^l_{(b)}\frac{\partial^2u^i_k}{\partial
x^j_{(a)}\partial v^l_{(b)}}+F^j_{(a)}F^l_{(b)}\frac{\partial^2 u^i_k}{\partial v^j_{(a)}\partial v^l_{(b)}
}\right)\right.\\
\!\!\!\! \! \left. +\sum_{a=1}^m\sum_{j=1}^n\left(F^j_{(a)}\frac{\partial u^i_k}{\partial
x^j_{(a)}}+\frac{\partial F^j_{(a)}}{\partial t}\frac{\partial u^i_k}{\partial
v^j_{(a)}}\right)+\sum_{a=1}^m\sum_{j,l=1}^n\left(v^l_{(a)}\frac{\partial
F^j_{(a)}}{\partial{x^l_{(a)}}}+F^l_{(a)}\frac{\partial F^j_{(a)}}{\partial{v^l_{(a)}}}\right)\frac{\partial
u^i_k}{\partial v^j_{(a)}}\right](p(t)).
\end{eqnarray*}
If we compare the above expression with
\begin{equation*}
\begin{aligned}(X_D^{(m)})_t^2u_k^i=\sum_{j,l=1}^n\sum_{a,b=1}^m\left(v^j_{(a)}\frac{\partial}{\partial {x^j_{(a)}}}+
F^{j}_{(a)}\frac{\partial}{\partial {v^{j}}_{(a)}}\right)\left(v^l_{(b)}\frac{\partial}{\partial {x^l_{(b)}}}+
F^{l}_{(b)}\frac{\partial}{\partial {v^{l}}_{(b)}}\right)u^i_k\\
\,\,\qquad\qquad=\sum_{j,l=1}^n\sum_{a,b=1}^m\left(v^j_{(a)}v^l_{(b)}\frac{\partial^2u^i_k}{\partial x^j_{(a)}\partial
x^l_{(b)}}+2v^j_{(a)}F^l_{(b)}\frac{\partial^2 u^i_k}{\partial x^j_{(a)}\partial
v^l_{(b)}}+F^j_{(a)}F^l_{(b)}\frac{\partial^2u^i_k}{\partial v^j_{(a)}\partial
v^l_{(b)}}\right)
\\\qquad\qquad\,\,\qquad\qquad+\sum_{a=1}^m\sum_{j=1}^nF^j_{(a)}\frac{\partial u^i_k}{\partial
x^{j}_{(a)}}+ \sum_{a=1}^m\sum_{j,l=1}^n\left(v^j_{(a)}\frac{\partial F^{l}_{(a)}}{\partial {x^j_{(a)}}
}+F^j_{(a)}\frac{\partial F^{l}_{(a)}}{\partial {v^j_{(a)}}}\right)\frac{\partial u^i_k}{\partial
v^{l}_{(a)}}\,,
\end{aligned}
\end{equation*}we obtain
\begin{equation}\label{eq2}
\frac{d^2\bar{x}^i}{dt^2}(t)=((X_D^{(m)})_t^2u_k^i+(X_L^{(m)})_tu_k^i)(p(t)),\qquad i=1,\ldots,n.
\end{equation}
As $\bar x(t)$ is a solution of system (\ref{SODE}), and in view of expressions (\ref{eq1}) and (\ref{eq2}), it
turns out  that
\begin{equation}\label{QuasiEq}
((X_D^{(m)})_t^2u_k^i+(X_L^{(m)})_tu_k^i)(p(t))=F^i\left(t,u_k(p(t)),((X_D^{(m)})_tu_k)(p(t))\right).
\end{equation}
Now, equation (\ref{QuasiEq}) holds for every fundamental system. This implies that, for each
$t\in\mathbb{R}$, the above equation remains valid for a generic  open and dense subset of $({\rm
T}\mathbb{R}^n)^m$. Hence,
$$
(X_D^{(m)})_t^2u^i_k+(X_L^{(m)})_tu^i_k=F^i\left(t,u_k,(X_D^{(m)})_tu_k\right),\qquad i=1,\ldots,n,
$$
for every $t\in\mathbb{R}$. Additionally, as the above procedure is still valid for every
$k\in\mathbb{R}^{2n}$, every superposition rule provides us with a family of $2n$-parametrized solutions
$u_k(\cdot)=\Upsilon(\cdot;k)$ of the $t$-parametrized family of systems of PDEs (\ref{Char}).

Consider now a fundamental system $x_{(1)}(t),\ldots,x_{(m)}(t)$ and denote $p=p(0)$. For an arbitrary
$(x_0,v_0)\in {\rm T}_{x_0}\mathbb{R}^n$, the theorem of the existence and uniqueness of solutions for systems of
first-order differential equations shows that there exists a solution $x(t)$ of system (\ref{SODE}) with
initial conditions $x(0)=x_0$ and $dx/dt(0)=v_0$. In view of the properties of superposition rules, there
exists a single $k\in\mathbb{R}^{2n}$ such that $x(t)=\Upsilon\left(p(t);k\right)=u_k\left(p(t)\right)$.
Consequently, in view of expression (\ref{eq1}), one has
\begin{equation*}
\left\{
 \begin{aligned}x^i_0&=u^i_k(p), \cr v^i_0&=\frac{du_k^i(p(t))}{dt}\bigg|_{t=0}=[(X_D^{(m)})_0u_k^i](p),\end{aligned} \qquad i=1,\ldots,n.\right.
\end{equation*}
In other words, for a generic $p\in({\rm T}\mathbb{R}^n)^m$, there exists a single $k\in\mathbb{R}^{2n}$ such
that $\varphi(p,k)=(x_0,v_0)$. It follows that $\varphi_p$ is a bijection that concludes the
``if'' part of our demonstration.

Let us now prove that a map $\Upsilon:({\rm T}\mathbb{R}^n)^m\times\mathbb{R}^{2n}\rightarrow\mathbb{R}^n$
satisfying conditions $(i)$ and $(ii)$ is a superposition rule for (\ref{SODE}). Consider any solution $x(t)$
of (\ref{SODE}). Given a generic family of $m$ particular solutions $x_{(1)}(t),\ldots,x_{(m)}(t)$ of
(\ref{SODE}), condition $(ii)$ ensures that there exists a unique $k\in\mathbb{R}^{2n}$ such that
$\varphi(p(0),k)=\left(x(0),dx/dt(0)\right)$, where we took again $p(t)$ to be of the form (\ref{Setting}). In
view of condition $(i)$, the function $u_k(\cdot)=\Upsilon(\cdot;k)$ is a solution for the family of systems
of PDEs (\ref{Char}). Defining now $\bar x(t)=u_k(p(t))$ and using that expressions (\ref{eq1}) and
(\ref{eq2}) are valid again, we get
\begin{equation}\label{SinSol}
\frac{d^2\bar{x}^i}{dt^2}(t)=F^i\left(t,\bar x(t),\frac{d\bar x}{dt}(t)\right),\qquad i=1,\ldots,n.
\end{equation}
That is, $\bar x(t)$ is a solution to (\ref{SODE}). Moreover, in view of condition $(ii)$ and formula
(\ref{eq1}),  $\bar x(0)=x(0)$ and $d\bar x/dt(0)=dx/dt(0)$. Consequently, $\bar x(t)$ and $x(t)$ are both
solutions of (\ref{SODE}) with the same initial conditions and they hence coincide. In summary, for every
solution $x(t)$ of system (\ref{SODE}) and a generic family of $m$ particular solutions, there exists a unique
$k\in\mathbb{R}^{2n}$ such that $x(t)=\Upsilon(p(t);k)$, so $\Upsilon$ is a superposition rule.
\end{proof}

Roughly speaking, Theorem \ref{MT} states that the existence of a superposition rule for a system  of SODEs
(\ref{SODE}) is determined by the existence of an `appropriate' $2n$-parametric family of particular solutions
of the family of systems of PDEs (\ref{Char}). The interest of this result is obvious: it characterizes not
only the existence of superposition rules for systems of SODEs, but also provides us with a tool, namely the
family of systems (\ref{Char}), to determine them.

\begin{Note} Note that $X_D$ and $X_L$ are properly defined $t$-dependent vector fields over $({\rm T}\mathbb{R}^{n})^m$
and they maintain the form (\ref{DefVec}) for every coordinate system on $({\rm T}\mathbb{R}^n)^m$ induced by a coordinate system on $\mathbb{R}^n$.
\end{Note}

\begin{Note} Denote by $S_{ij}$ a permutation of variables $x_{(i)}\leftrightarrow x_{(j)}$, with $i,j=1,\ldots,m$.
As (\ref{Char}) and (\ref{Fun}) are invariant under such permutations, it can be easily inferred that, if $\Upsilon$
is a superposition rule for (\ref{SODE}), then $S_{ij}\Upsilon$ is also, which provides an analogue for systems of SODEs
of a known result about standard superposition rules \cite{CGM07}.
\end{Note}

Apart from the main result of Theorem \ref{MT}, a careful analysis of its proof suggests us new types of
superposition rules for systems of SODEs generalising previous notions used in the study of first-order
differential equations \cite{CGM07,In72}. Indeed, given a particular solution of the family of systems of PDEs
(\ref{Char}) and a family of particular solutions $x_{(1)}(t),\ldots,x_{(m)}(t)$ of system (\ref{SODE}), we
can define
$$
\bar x(t)=u_k\left(x_{(1)}(t),\frac{dx_{(1)}}{dt}(t),\ldots,x_{(m)}(t),\frac{dx_{(m)}}{dt}(t)\right).
$$
The above expression has the same form as (\ref{NewSol}). Following the calculations carried out in the ``if''
part of Theorem \ref{MT}, we obtain that the first and the second derivative of the above curve satisfy relations
(\ref{eq1}) and (\ref{eq2}). From here, as $u_k$ is a solution of (\ref{Char}), it follows that $\bar x(t)$ is
a new solution of (\ref{SODE}). In other words, a particular solution of the systems of PDEs (\ref{Char})
allows us to generate new solutions of system (\ref{SODE}) from any set of $m$ particular solutions for this
same system. This fact enables us to define a new type of superposition rule for systems of SODEs as follows.

\begin{definition} A {\it partial superposition rule} for a system of SODEs (\ref{SODE})
is a mapping  $\mathcal{P}:({\rm T}\mathbb{R}^n)^m\times\mathbb{R}^{p}\rightarrow \mathbb{R}^n$, with $p<2n$, such that
\begin{itemize}
 \item For a generic set $x_{(1)}(t),\ldots,x_{(m)}(t)$ of  particular solutions of system (\ref{SODE}),
$$
\bar
x(t)=\mathcal{P}\left(x_{(1)}(t),\frac{dx_{(1)}}{dt}(t),\ldots,x_{(m)}(t),\frac{dx_{(m)}}{dt}(t);k_1,\ldots,k_p\right)
$$
is a new solution of (\ref{SODE}).
\item For a generic $p\in ({\rm T}\mathbb{R}^n)^m$, the map $\mathcal{P}_p:\bar k\in\mathbb{R}^{p}\mapsto \mathcal{P}(p;\bar k)\in\mathbb{R}^n$
is an immersion.
\end{itemize}
\end{definition}
Obviously, for every fixed $\bar k=(k_1,\ldots,k_p)$, the map $u_{\bar k}(\cdot)=\mathcal{P}(\,\cdot\,;\bar
k)$ is a solution of the system (\ref{Char}). In view of this, it is easy to generalize Theorem \ref{MT} in
order to characterize systems of SODEs admitting partial superposition rules. Moreover, the above notion
extends to systems of SODEs the notion of partial superposition rule for systems of first-order differential
equations defined in \cite{CGM07}.

Let us now illustrate how the above results and definitions work. Consider the SODE
\begin{equation}\label{lesssimple}
\frac{d^2x}{dt^2}=t^2
\end{equation}
and look for a superposition rule depending on a single particular solution. Following the terminology used in
Theorem \ref{MT}, we have $m=1$ (one particular solution) and $n=1$ (system defined on $\mathbb{R}$).
Consequently, the corresponding family of systems of PDEs (\ref{Char}) reads
\begin{equation}\label{fam2}
v^2\frac{\partial^2 u}{\partial x_{(1)}^2}+2v_{(1)}t^2\frac{\partial^2 u}{\partial x_{(1)}\partial
v_{(1)}}+t^4\frac{\partial^2 u}{\partial v_{(1)}^2}+t^2\frac{\partial u} {\partial x_{(1)}}+2t \frac{\partial
u}{\partial v_{(1)}}=t^2,
\end{equation}
whose common solutions, which do not depend on $t$, are solutions of the system
$$
\frac{\partial^2 u}{\partial v_{(1)}^2}=0,\qquad \frac{\partial u}{\partial v_{(1)}}=0,\qquad
2v_{(1)}\frac{\partial^2 u}{\partial x_{(1)}\partial v_{(1)}}+\frac{\partial u} {\partial x_{(1)}}=1,\qquad
\frac{\partial ^2u}{\partial x_{(1)}^2}=0.
$$
The solutions of the above system are of the form $u=x_{(1)}+k_1$, with $k_1$ being an arbitrary constant.
Obviously, the family of systems (\ref{fam2}) does not give rise to a two parametric family of solutions and
(\ref{lesssimple}) does not admit any superposition rule in terms of one particular solution. Nevertheless, it
is interesting to point out that the solutions $u=x_{(1)}+k_1$ exemplify that, for every particular solution
$x_{(1)}(t)$ of (\ref{fam2}), the new function $u(x_{(1)}(t))=x_{(1)}(t)+k_1$ is a new solution of the system
that gives rise to a partial superposition rule $\mathcal{P}:(x_{(1)},v_{(1)};k_1)\in{\rm
T}\mathbb{R}\times\mathbb{R}\mapsto x_{(1)}+k_1\in\mathbb{R}$ for equation (\ref{lesssimple}).

Let us now turn to determining all possible superposition rules for (\ref{lesssimple}) involving two particular
solutions.  So, we have $m=2,$ $n=1$, and  the family (\ref{Char}) reads
\begin{equation*}
\begin{aligned}\sum_{a,b=1}^2\!\left(v_{(a)}v_{(b)}\frac{\partial^2 u}{\partial x_{(a)}\partial
x_{(b)}}\!+2v_{(a)}t^2\frac{\partial^2 u}{\partial x_{(a)}\partial v_{(b)}}\!+t^4\frac{\partial^2 u}{\partial
v_{(a)}\partial v_{(b)} }\right)\!+\!\sum_{a=1}^2\!\left(t^2\frac{\partial u}{\partial x_{(a)}}\!+\!2t\frac{\partial
u}{\partial v_{(a)}}\!\right)\!=\!t^2.
\end{aligned}
\end{equation*}
Proceeding as before, we obtain that solutions of this $t$-parametrized family of PDEs are solutions of the
system
\begin{eqnarray*}
\sum_{a=1}^2\frac{\partial
u}{\partial v_{(a)}}=0,\quad \sum_{a,b=1}^2\frac{\partial^2 u}{\partial v_{(a)}\partial
v_{(b)}}=0, \quad \sum_{a,b=1}^2v_{(a)}v_{(b)}\frac{\partial^2 u}{\partial x_{(a)}\partial
x_{(b)}}=0,\\ \sum_{a,b=1}^22v_{(a)}\frac{\partial^2 u}{\partial x_{(a)}\partial v_{(b)}}+\sum_{a=1}^2\frac{\partial u}{\partial x_{(a)}}=1.
\end{eqnarray*}
Plugging the first equation of the above system into the others, we obtain that the above system is equivalent
to
\begin{equation*}
\frac{\partial u}{\partial v_{(1)}}=-\frac{\partial u}{\partial v_{(2)}},\quad \frac{\partial u}{\partial
x_{(1)}}=1-\frac{\partial u}{\partial x_{(2)}},\quad (v_{(1)}-v_{(2)})^2\frac{\partial^2 u}{\partial
x_{(1)}\partial x_{(2)}}=0,
\end{equation*}
whose solutions take the form $u=(x_{(1)}-x_{(2)})f_1+x_{(1)}+f_2$,
with  $f_1$ and $f_2$ being two arbitrary functions depending on $v_{(1)}-v_{(2)}$. Now, by choosing
appropriate one-parametric families of solutions of the above form, we can get partial superposition rules.
For example, setting $f_1=k$ and $f_2=0$, with $k\in\mathbb{R}$, we obtain the family of solutions
$u_k=(x_{(1)}-x_{(2)})k+x_{(1)}$ which results in the partial superposition rule
$\mathcal{P}(x_{(1)},x_{(2)};k)=(x_{(1)}-x_{(2)})k+x_{(1)}$, which generates new particular solutions out of two
known ones and one constant. Moreover, Theorem \ref{MT} shows that the determination of a superposition rule
for system (\ref{lesssimple}) amounts us to obtaining a two-parametric family of solutions $u_{(k_1,k_2)}$ of the
above form such that condition $(ii)$ of Theorem \ref{MT} holds. This can be done in several ways. For
instance, by setting $f_1=k_1$ and $f_2=k_2$, with $k_1,k_2\in\mathbb{R}$, we obtain the superposition rule
$$
\Upsilon(x_1,v_1,x_2,v_2;k_1,k_2)=u_{(k_1,k_2)}(x_{(1)},x_{(2)})=k_1(x_{(1)}-x_{(2)})+x_{(1)}+k_2,
$$
and if we choose $f_1=k_1(v_{(1)}-v_{(2)})$ and $f_2=k_2(v_{(1)}-v_{(2)})$, we arrive to
$$
\Upsilon(x_1,v_1,x_2,v_2;k_1,k_2)=u_{(k_1,k_2)}(x_{(1)},x_{(2)})=k_1(v_{(1)}-v_{(2)})(x_{(1)}-x_{(2)})+x_{(1)}+k_2(v_{(1)}-v_{(2)}).
$$

Using our methods, we can easily derive the results of Table \ref{tab:template}. Special attention must be
paid to the second example, illustrating that partial superposition rules may exist when superposition rules
depending on the same number of particular solutions do not. In addition, this particular example is not a
SODE Lie system (see definition in the next section), which was almost the only tool to study superposition rules
for systems of SODEs so far.

\begin{table}[h]
\centering
\begin{tabular}{| c | c | c | }
    \hline
    SODE & Superposition rule & Partial superposition rule\\[.05in] \hline
    $ \frac{d^2x}{dt^2}=0$ & $v_{(1)}(k_1x_{(1)}+k_2)$ & $x_{(1)}+k_1$ \\[.05in] \hline
    $\frac{d^2x}{dt^2}=t\left(\frac{dx}{dt}\right)^2$ & nonexistent& $x_{(1)}+k_1$ \\[.05in] \hline
$\frac{d^2x}{dt^2}=t^2\frac{dx}{dt}$ & $k_1x_{(1)}+k_2$ & $k_1x_{(1)}$ \\[.05in] \hline
   \end{tabular}
\caption{Superposition and partial superposition rules depending on a particular solution.} \label{tab:template}
\end{table}

Theorem \ref{MT} characterizes systems of SODEs possessing a base-superposition rule as follows.
\begin{corollary} A mapping $\Upsilon:(\mathbb{R}^n)^m\times\mathbb{R}^{2n}\rightarrow\mathbb{R}^n$ is a
base-superposition rule for a system of SODEs (\ref{SODE}) if and only if the following two conditions hold:
\begin{enumerate}
 \item The functions $u_k:p\in\mathbb{R}^{nm}\mapsto u_k(p)=\Upsilon(p;k)\in\mathbb{R}^n$, with $k\in\mathbb{R}^{2n}$,
 are solutions of the $t$-parametrized family of systems of PDEs on $({\rm T}\mathbb{R}^n)^m$ given by
\begin{equation*}
({X}^{(m)}_1)_t^2u^i_k+(X^{(m)}_2)_tu^i_k=F^i(t,u,(X^{(m)}_1)_tu_k),\qquad i=1,\ldots,n,
\end{equation*}
where $X^{(m)}_1, X^{(m)}_2$ are the diagonal prolongations to $({\rm T}\mathbb{R}^n)^m$ of the time-dependent
vector fields
$$
X_1=\sum_{i=1}^nv^i\frac{\partial}{\partial x^i}, \qquad X_2=\sum_{i=1}^nF^i(t,x,v)\frac{\partial}{\partial
x^i}.
$$
\item The map $\varphi:({\rm T}\mathbb{R}^n)^m\times\mathbb{R}^{2n}\rightarrow{\rm T}\mathbb{R}^n$ of the form
$$
\varphi(p;k)=(u_k(p),[(X_1^{(m)})_0u_k](p))\in {\rm T}_{u_k(p)}\mathbb{R}^n
$$
gives rise to a family of bijections $\varphi_p:k\in\mathbb{R}^{2n}\mapsto \varphi(p;k)\in{\rm
T}\mathbb{R}^{n}$, with $p$ being a generic point of $({\rm T}\mathbb{R}^n)^m$.
\end{enumerate}
\end{corollary}

\section{Superposition rules and SODE Lie systems}
\setcounter{equation}{0}
Recently, the theory of Lie systems was employed to obtain a few results about superposition rules for systems
of SODEs \cite{RSW97}--\cite{SIGMA}. All these achievements were based on the notion of a {\it SODE Lie system}.
We now describe this concept and provide several new results about the use of Lie systems to analyse
different types of superposition rules for systems of SODEs.

\begin{definition}\label{DefiSODE} A system of second-order ordinary differential equations (\ref{SODE})
is a {\it SODE Lie system} if the first-order system
\begin{equation}\label{FirstOrd}
\left\{\begin{aligned}\frac{dx^i}{dt}&=v^i,\cr
\frac{dv^i}{dt}&=F^i(t,x,v),\end{aligned}\right. \qquad i=1,\ldots,n,
\end{equation}
obtained from (\ref{SODE}) by adding the variables $v^i\equiv dx^i/dt$, $i=1,\ldots,n$, is a Lie
system.
\end{definition}

The Lie--Scheffers theorem is an effective tool to determine whether a system (\ref{SODE}) is a SODE Lie system or not.
Nevertheless, this method is based on analysing properties of the time-dependent vector field associated with
the corresponding system (\ref{FirstOrd}) and it does not provide any straightforward information about the
superposition rules for these systems. In order to overcome this drawback, we provide the following
characterization of SODE Lie systems in terms of properties of superposition rules.

\begin{proposition}\label{SODEsSup} A system of SODEs (\ref{SODE}) is a SODE Lie system if and only if it admits
a superposition rule $\Upsilon:({\rm T}\mathbb{R}^n)^m\times\mathbb{R}^{2n}\rightarrow\mathbb{R}^n$ such that $X_L^{(m)}\Upsilon=0$,
where $X_L$ is given by (\ref{DefVec}).
\end{proposition}

\begin{proof}
Consider a second-order system of the form (\ref{SODE}) admitting a superposition rule (\ref{SupSODE1}). The
general solution $x(t)$ of this system can be put in the form (\ref{SupSODE}). Differentiating this expression
with respect to the time, we obtain
\begin{equation}\label{ex1}
 \frac{dx^i}{dt}(t)=\sum_{a=1}^m\sum_{j=1}^n\frac{dx^j_{(a)}}{dt}(t)\frac{\partial\Upsilon^i}{\partial
x^j_{(a)}}(p(t))+\sum_{a=1}^m\sum_{j=1}^nF^j\left(t,x_{(a)}(t),\frac{dx_{(a)}}{dt}(t)\right)\frac{\partial\Upsilon^i}{\partial
v^j_{(a)}}(p(t)),
\end{equation}
where $i=1,\ldots,n$ and $p(t)$ is given by (\ref{Setting}). Therefore, by defining
\begin{equation}\label{ex4}
 \widehat\Upsilon^i(t,x_{(1)},v_{(1)},\ldots,x_{(m)},v_{(m)})=\sum_{a=1}^m\sum_{j=1}^nv^j_{(a)}\frac{\partial\Upsilon^i}{\partial
x^j_{(a)}}+\sum_{a=1}^m\sum_{j=1}^nF^j\left(t,x_{(a)},v_{(a)}\right)\frac{\partial\Upsilon^i}{\partial
v^j_{(a)}},
\end{equation}
expressions (\ref{SupSODE}) and (\ref{ex1}) can be brought into the form
\begin{equation}\label{sys}
\left\{\begin{aligned}x^i(t)&=\Upsilon^i\left(x_{(1)}(t),\frac{dx_{(1)}}{dt}(t),\ldots,x_{(m)}(t),\frac{dx_{(m)}}{dt}(t);k_1,\ldots,k_{2n}\right),\cr
\frac{dx^i}{dt}(t)&=\widehat\Upsilon^i\left(t,x_{(1)}(t),\frac{dx_{(1)}}{dt}(t),\ldots,x_{(m)}(t),\frac{dx_{(m)}}{dt}(t);k_1,\ldots,k_{2n}\right),\end{aligned}\right.\,\, i=1,\ldots,n.
\end{equation}
Taking into account that the general solution $(x(t),v(t))$ of the first-order system (\ref{FirstOrd}) is
obtained by adding the variables $v^i=dx^i/dt$ to system (\ref{SODE}), we see that expressions (\ref{sys})
define a map $\bar\Phi:(t,p;k)\in\mathbb{R}\times ({\rm T}\mathbb{R}^n)^m\times\mathbb{R}^{2n}\mapsto
(\Upsilon(p;k),\widehat\Upsilon(t,p;k))\in {\rm T}\mathbb{R}^n$ which allows us to write the general solution
of this first-order system in terms of a generic set of particular solutions $(x_{(a)}(t),v_{(a)}(t))$, with
$a=1,\ldots,m$. In view of expression (\ref{ex4}),
$$
\frac{\partial \widehat{\Upsilon}^i}{\partial
t}(t,x_{(1)},v_{(1)},\ldots,x_{(m)},v_{(m)})=\sum_{a=1}^m\sum_{j=1}^n\frac{\partial F^j}{\partial
t}\left(t,x_{(a)},v_{(a)}\right)\frac{\partial\Upsilon^i}{\partial v^j_{(a)}}=X^{(m)}_L\Upsilon^i.
$$
Therefore, if $X^{(m)}_L\Upsilon=0$, the mapping $\widehat\Upsilon$ and, in consequence, $\bar\Phi$ are
time-independent. This shows that the function $\bar \Phi$ is a superposition rule for system
(\ref{FirstOrd}), which is therefore a Lie system. Hence, system (\ref{SODE}) is a SODE Lie system.

Let us assume now that system (\ref{SODE}) is a SODE Lie system, i.e. the first-order system (\ref{FirstOrd})
is a Lie system and there exists a superposition rule $\Phi:(p;k)\in{\rm
T}\mathbb{R}^{nm}\times\mathbb{R}^{2n}\mapsto (\Upsilon(p;k),{\Phi}_v(p;k))\in{\rm T}\mathbb{R}^n$ such that
its general solution $(x(t),v(t))$ can be written as
\begin{equation*}
\left\{\begin{aligned}x(t)=\Upsilon(x_{(1)}(t),v_{(1)}(t),\ldots,x_{(m)}(t),v_{(m)}(t);k_1,\ldots,k_{2n}),\cr
v(t)={\Phi}_v(x_{(1)}(t),v_{(1)}(t),\ldots,x_{(m)}(t),v_{(m)}(t);k_1,\ldots,k_{2n}),\end{aligned}\right.
\end{equation*}
where $(x_{(a)}(t),v_{(a)}(t))$, with $a=1,\ldots,m$, is a generic family of particular solutions of system
(\ref{FirstOrd}). Since $dx_{(a)}(t)/dt=v_{(a)}(t)$, the function $\Upsilon:{\rm
T}\mathbb{R}^{nm}\times\mathbb{R}^{2n}\rightarrow\mathbb{R}^n$ enables us to write the general solution $x(t)$
of system (\ref{SODE}) in the form
$$
x(t)=\Upsilon\left(x_{(1)}(t),\frac{dx_{(1)}}{dt}(t),\ldots,x_{(m)}(t),\frac{dx_{(m)}}{dt}(t);k_1,\ldots,k_{2n}\right),\\
$$
i.e. in terms of a generic family of its particular solutions $x_{(1)}(t),\ldots,x_{(m)}(t)$, their
derivatives, and a set of $2n$-constants. In other words, our system of SODEs admits a superposition rule.
Consequently, differentiating the above expression with respect to $t$, we obtain, in virtue of
(\ref{ex4}), that $\bar{\Phi}_v\left(p(t);k\right)=
\widehat\Upsilon\left(t,p(t);k\right)$ for a generic $p(t)$, which is given by (\ref{Setting}) and constructed from a family of particular solutions $x_{(1)}(t),\ldots,x_{(m)}(t)$. Hence,
$\bar{\Phi}_v\left(x_{(1)},v_{(1)},\ldots,x_{(m)},v_{(m)};k\right)=\widehat\Upsilon\left(t,x_{(1)},v_{(1)},\ldots,x_{(m)},v_{(m)};k\right)$ and
$$
\frac{\partial\bar{\Phi}^i_v}{\partial
t}\left(x_{(1)},v_{(1)},\ldots,x_{(m)},v_{(m)};k\right)=\frac{\partial\widehat{\Upsilon}^i}{\partial
t}\left(t,x_{(1)},v_{(1)},\ldots,x_{(m)},v_{(m)};k\right)=X_L^{(m)}\Upsilon^i=0,
$$
for $i=1,\ldots,n$.
\end{proof}

The above proposition improves the results of \cite{CL10SecOrd2}, where it is only stated that SODE Lie
systems admit superposition rules. Indeed, our new result also supplies additional information about such
superposition rules, namely, $X^{(m)}_L\Upsilon=0$. This property is going to be used next to retrieve easily
some previous results found in \cite{CL10SecOrd2} and various new ones.

\begin{proposition}\label{QuasiBase} Every system of SODEs (\ref{SODE}) admitting a quasi-base-superposition rule is a SODE Lie system.
\end{proposition}
\begin{proof}
Assume that system (\ref{SODE}) admits a quasi-base-superposition rule. Let us now prove that this quasi-base
superposition rule gives rise to a superposition rule for (\ref{SODE}) such that $X_L^{(m)}\Upsilon=0$, which,
in view of Proposition \ref{SODEsSup}, proves that system (\ref{SODE}) is a SODE Lie system.

The general solution $x(t)$ of (\ref{SODE}) can be cast in the form (\ref{express1}). As the functions
$I_j:{\rm T}\mathbb{R}^{nm}\rightarrow \mathbb{R}$, with $j=1,\ldots,q$, are constant along the $m$-tuples
$(x_{(a)}(t),dx_{(a)}(t)/dt)$ obtained from $m$ particular solutions $x_{(a)}(t)$ of system (\ref{SODE}), i.e.
\begin{equation}\label{exp1}
I_j\left(x_{(1)}(t),\frac{dx_{(1)}}{dt}(t),\ldots,x_{(m)}(t),\frac{dx_{(m)}}{dt}(t)\right)={\rm const.,}
\end{equation}
we obtain
$$
\frac{dI_j}{dt}=\sum_{i=1}^n\sum_{a=1}^m\left[\frac{dx^i_{(a)}}{dt}(t)\frac{\partial I_j}{\partial
x^i_{(a)}}(p(t))+F^i\left(t,x_{(a)}(t),\frac{dx_{(a)}}{dt}(t)\right)\frac{\partial I_j}{\partial
v^i_{(a)}}(p(t))\right]=0,
$$
where $p(t)$ is given by expression (\ref{Setting}). The above holds for every generic family of particular
solutions. Then, $X^{(m)}_DI_j=0$, for $j=1,\ldots,q$. Substituting expression (\ref{exp1}) into
(\ref{express1}), it turns out that there exists a superposition rule  $\Upsilon:({\rm
T}\mathbb{R}^n)^m\times\mathbb{R}^{2n}\rightarrow \mathbb{R}^{n}$ for system (\ref{SODE}) of the form
$\Upsilon\left(p;k\right)=G(x_{(1)},\ldots,x_{(m)},I_1(p),\ldots,I_{q}(p);k)$, where
$p=(x_{(1)},v_{(1)},\ldots,x_{(m)},v_{(m)})$. Indeed, in view of the definition of $\Upsilon$ and the
properties of the quasi-base-superposition rule $G$,
$$x(t)=G(x_{(1)}(t),\ldots,x_{(m)}(t),I_1(p(t)),\ldots,I_{q}(p(t));k)=\Upsilon\left(p(t);k\right),$$
where $x_{(1)}(t),\ldots,x_{(m)}(t)$ is any generic family of particular solutions of system (\ref{SODE}).
Then,
$$
X_L^{(m)}\Upsilon=\sum_{j=1}^{q}\sum_{i=1}^n\sum_{a=1}^m\frac{\partial F^i}{\partial
t}\left(t,x_{(a)},\frac{dx_{(a)}}{dt}\right)\frac{\partial I_j}{\partial v^i_{(a)}}\frac{\partial G}{\partial
I_j}=\sum_{j=1}^{q}\frac{\partial}{\partial t}(X^{(m)}_DI_j)\frac{\partial G}{\partial I_j}=0.
$$
\end{proof}

\begin{corollary}\label{Free} Every system of SODEs admitting
 a base-superposition rule is a SODE Lie system.
\end{corollary}

The implication of the above corollary cannot be reversed, i.e. not every SODE Lie system admits a
base-superposition rule. Indeed, the following results can be easily used to prove the existence of SODE Lie
systems admitting no base-superposition rule.

\begin{lemma}\label{CondNec} Given a system of SODEs (\ref{SODE}) admitting a base-superposition rule, the systems
\begin{equation}\label{FamSODE}
\frac{d^2x^i}{d\tau^2}=\frac{d\tau}{dt}\frac{d^2t}{d\tau^2}\frac{dx^i}{d\tau}+\left(\frac{dt}{d\tau}\right)^2F^i
\left(t(\tau),x,\frac{d\tau}{dt}\frac{dx}{d\tau}\right),\qquad
i=1,\ldots,n,
\end{equation}
with $t=t(\tau)$ being any time-reparametrisation, are SODE Lie systems admitting a common base-superposition
rule.
\end{lemma}
\begin{proof} A time-reparametrisation $t=t(\tau)$ maps system (\ref{SODE}) to a system of SODEs of the
form (\ref{FamSODE}) with the general solution $x(\tau)=x(t(\tau))$. If $\Upsilon:\mathbb{R}^{nm}\times\mathbb{R}^{2n}\rightarrow\mathbb{R}^n$
is a base-superposition rule for (\ref{SODE}),  then $
x(\tau)=\Upsilon(x_{(1)}(\tau),\ldots,x_{(m)}(\tau);k_1,\ldots,k_{2n}), $ where
$x_{(1)}(\tau),\ldots,x_{(m)}(\tau)$ is any generic family of particular solutions of (\ref{FamSODE}).
Consequently, all the second-order differential equations of the family (\ref{FamSODE}) admit a common
base-superposition rule and, according to Corollary \ref{Free}, all systems (\ref{FamSODE}) are SODE Lie
systems.
\end{proof}

Although only few SODE Lie systems admit them, base-superposition rules are the main superposition rules treated in the literature. The next proposition shows that SODE Lie systems must satisfy various restrictive
conditions to admit a base-superposition rule.

\begin{proposition}\label{IncField} Given a system of SODEs (\ref{SODE}) admitting a base-superposition rule,
the associated first-order system (\ref{FirstOrd}) is a Lie system related to a Vessiot--Guldberg Lie algebra
containing the Liouville vector field $\Delta_L$ of the tangent bundle ${\rm T}\mathbb{R}^n$ and the vector fields
\begin{equation}\label{IncFieldVec}
X^\lambda_p(x,v)=\sum_{i=1}^n \frac{d^p}{d\lambda^p} (\lambda^{2}F^i(t_1,x,\lambda^{-1} v))\frac{\partial}{\partial
v^i},\qquad \qquad p=1,2,\ldots,
\end{equation}
where $t_1\in\mathbb{R}$ and $\lambda\in\mathbb{R}^\times=\mathbb{R}-\{0\}$.
\end{proposition}
\begin{proof} In view of Lemma \ref{CondNec}, time-reparametrisations $\tau=\tau(t)$, with inverses $t=t(\tau)$,
transform system (\ref{SODE}) into  the family of SODE Lie systems (\ref{FamSODE}), whose associated first-order systems
\begin{equation*}
\left\{\begin{aligned}\frac{dx^i}{d\tau}&=v^i,\\
\frac{dv^i}{d\tau}&=\frac{d\tau}{dt}\frac{d^2t}{d\tau^2}v^i+\left(\frac{dt}{d\tau}\right)^2F^i\left(t(\tau),x,\frac{d\tau}{dt}v\right),\end{aligned}\right.
\end{equation*}
with $i=1,\ldots,n$, admit a common base-superposition rule. According to Proposition \ref{CSR}, this implies that there exists a
finite-dimensional Lie algebra of vector fields $V$ containing all the vector fields
$$
X^{ t(\tau)}_\tau(x,v)=\sum_{i=1}^n\left( v^i\frac{\partial}{\partial
x^i}+\left(\frac{d\tau}{dt}\frac{d^2t}{d\tau^2}v^i+\left(\frac{dt}{d\tau}\right)^2F^i\left(t(\tau),x,
\frac{d\tau}{dt}v\right)\right)\frac{\partial}{\partial
v^i}\right),
$$
with $t=t(\tau)$ being any time-reparametrisation and $\tau\in\mathbb{R}$. In particular, if we take
$t(\tau)=\tau$, we obtain that
$$
X_\tau=\sum_{i=1}^n\left( v^i\frac{\partial}{\partial x^i}+F^i\left(\tau,x,v\right)\frac{\partial}{\partial
v^i}\right)\in V,\qquad \forall\tau\in\mathbb{R}.
$$
Now, if we take a time-reparametrisation $t=\bar t(\tau)$ such that
$$\bar t(\tau_1)=t_1,\qquad \,\,\frac{d\bar t}{d\tau}(\tau_1)=1,\qquad \,\, \frac{d^2\bar t}{d\tau^2}(\tau_1)=1,$$
it follows that $X^{\bar t(\tau)}_{\tau_1},X_{t_1}\in V$ and, consequently, $X^{\bar
t(\tau)}_{\tau_1}-X_{t_1}\in V$. Taking into account that
$$
X^{\bar t(\tau)}_{\tau_1}(x,v)=\sum_{i=1}^n\left( v^i\frac{\partial}{\partial
x^i}+\left(v^i+F^i\left(t_1,x,v\right)\right)\frac{\partial}{\partial v^i}\right),
$$
we have $X^{\bar t(\tau)}_{\tau_1}-X_{t_1}=\sum_{i=1}^nv^i\partial/\partial v^i=\Delta_L\in V$.

On the other hand, consider a family of parametrizations $t=t_{\lambda,t_1}(\tau)$, with
$\lambda\in\mathbb{R}^\times$ and $t_1\in\mathbb{R}$, satisfying the conditions
$$t_{\lambda,t_1}(\tau_1)=t_1,\qquad \,\,\frac{d t_{\lambda,t_1}}{d\tau}(\tau_1)=\lambda,\qquad \,\,\frac{d^2 t_{\lambda,t_1}}{d\tau^2}(\tau_1)=0.$$
Consequently, the family of vector fields
$$
X^\lambda_{\tau_1}(x,v)\equiv X^{t_{\lambda,t_1}(\tau_1)}_{\tau_1}(x,v)=\sum_{i=1}^n\left(
v^i\frac{\partial}{\partial x^i}+\lambda^{2}F^i\left(t_1,x,\lambda^{-1} v\right)\frac{\partial}{\partial
v^i}\right),\qquad \lambda\in\mathbb{R}^\times,
$$
is included in $V$. Note that, for every $t_1$, the above family of vector fields can be considered as a curve
in $V$. As $V$ is a vector space, all the derivatives of such a curve, that is, the vector fields
$$
\frac{d^p}{d\lambda^p}\left[X^\lambda_{\tau_1}(x,v)\right]=\sum_{i=1}^n \frac{d^p}{d\lambda^p}
(\lambda^{2}F^i(t_1,x,\lambda^{-1} v))\frac{\partial}{\partial v^i}=X^\lambda_p(x,v),
$$
are included in $V$.
\end{proof}

\section{Superposition rules and systems of HODEs}\label{HODEs}
\setcounter{equation}{0}
In order to introduce a general theory of superposition rules for systems of HODEs, let us recall some basic concepts of
the theory of higher-order tangent bundles \cite{LOS01}.

Given two curves $\rho,\sigma:\mathbb{R}\rightarrow\mathbb{R}^n$ such that
$\rho(0)=\sigma(0)=x_0\in\mathbb{R}^n$, we say that they have a {\it contact of order $s$ at $x_0$}, with
$s\in\mathbb{N}$,  if they satisfy
$$
\frac{d^j(f\circ\rho)}{dt^j}(0)=\frac{d^j(f\circ\sigma)}{dt^j}(0),\qquad j=1,\ldots,s,
$$
for every function $f\in C^{\infty}(\mathbb{R}^n)$. The relation `to have a contact of order $s$ at $x_0$' is
an equivalence relation. Each equivalence class, say ${\bf t}^s_{x_0}$, is called an {\it $s$-tangent vector}
at $x_0$. Now, we define ${\rm T}^s_{x_0}\mathbb{R}^n$ as the set of all $s$-tangent vectors at $x_0$ and we
put
$$
{\rm T}^s\mathbb{R}^n=\bigcup_{x_0\in\mathbb{R}^n}{\rm T}^s_{x_0}\mathbb{R}^n.
$$
It can be proved that $({\rm T}^s\mathbb{R}^n,\pi,\mathbb{R}^n)$, with $\pi:{\bf t}^s_{x_0}\in{\rm
T}^s\mathbb{R}^n\mapsto x_0\in\mathbb{R}^n$, can be endowed with a differential structure of fibre bundle. Let
us briefly analyse this fact.

Every global coordinate system $\{x^1,\ldots,x^n\}$ on $\mathbb{R}^n$ induces a natural coordinate system on
the space ${\rm T}^s\mathbb{R}^n$. Indeed, consider again a curve $\rho$. The $s$-tangent vector ${\bf
t}^s_{x_0}$ associated with this curve admits a representative
$$
\rho^i(0)+\frac{t}{1!}\frac{d\rho^i}{dt}(0)+\ldots+\frac{t^s}{s!}\frac{d^s\rho^i}{dt^s}(0),\qquad
i=1,\ldots,n,
$$
which can be characterized by its coefficients
$$
x_0^i=\rho^i(0),\quad y^{(1)i}_0=\frac{1}{1!}\frac{d\rho^i}{dt}(0),\quad \ldots,\quad
y^{(s)i}_0=\frac{1}{s!}\frac{d^s\rho^i}{dt^s}(0),\qquad i=1,\ldots,n.
$$
In consequence, the mapping $\varphi:{\bf t}^s_{x_0}\in{\rm T}^s\mathbb{R}^n\mapsto
(x^i_0,y^{(1)i}_0,\ldots,y^{(s)i}_0)\in \mathbb{R}^{(s+1)n}$ gives a canonical global coordinate for ${\rm
T}^s\mathbb{R}^n$. Obviously, the map $\pi$ becomes a smooth submersion which makes ${\rm T}^s\mathbb{R}^n$
into a fibre bundle with base $\mathbb{R}^n$. We hereby denote each element of ${\rm T}^s\mathbb{R}^n$ by
${\bf t}^s_x=(x,y^{(1)},\ldots,y^{(s)})$.

Now, given a curve $c:t\in \mathbb{R}\mapsto c(t)\in\mathbb{R}^n$, we call {\it prolongation to ${\rm
T}^s\mathbb{R}^n$ of $c$} the curve ${\bf t}^sc:t\in \mathbb{R} \mapsto {\bf t}^sc(t)\in {\rm
T}^s\mathbb{R}^n$, associating with every $t_0$ the corresponding equivalence class of $c(t-t_0)$, and given
in coordinates by
$$
{\bf t}^sc(t)=\left(x(t),\frac{1}{1!}\frac{dc}{dt}(t),\ldots,\frac{1}{s!}\frac{d^sc}{dt^s}(t)\right).
$$

\begin{definition}\label{Def1b} We say that a system of $s$-order ordinary differential equations on $\mathbb{R}^n$ given by
\begin{equation}\label{HODE}
\frac{d^sx^i}{dt^p}=F^i\left(t,x,\frac{dx}{dt},\ldots,\frac{d^{s-1}x}{dt^{s-1}}\right), \qquad i=1,\ldots,n,
\end{equation}
admits a {\it superposition rule} if there exists a map $\Upsilon:({\rm
T}^{s-1}(\mathbb{R}^{n}))^m\times\mathbb{R}^{sn}\rightarrow \mathbb{R}^n$ of the form
\begin{equation*}
x=\Upsilon({\bf t}^{s-1}_{x_{(1)}},\ldots,{\bf t}^{s-1}_{x_{(m)}};k_1,\ldots,k_{sn}),
\end{equation*}
such that the general solution $x(t)$ of system (\ref{HODE}) can be written as
\begin{equation}\label{SupHODE}
x(t)=\Upsilon\left({\bf t}^{s-1}{x_{(1)}}(t),\ldots,{\bf t}^{s-1}{x_{(m)}}(t);k_1,\ldots,k_{sn}\right),
\end{equation}
for any generic family of particular solutions $x_{(1)}(t),\ldots, x_{(m)}(t)$ of (\ref{HODE}) and
$k_1,\ldots,k_{sn}$ being a set of constants related to the initial conditions of each particular solution.
\end{definition}

\begin{Note} Observe that, according to the above definitions, we have ${\rm T}^1\mathbb{R}^n\simeq {\rm T}\mathbb{R}^n$
and ${\bf t}^1x_{(a)}(t)=\left(x_{(a)}(t),dx_{(a)}(t)/dt\right)$. Hence, Definition \ref{Def1} describing superposition
rules for systems of SODEs turns out to be a particular case of the above definition. Moreover,
if we put ${\rm T}^0\mathbb{R}^n=\mathbb{R}^n$,  Definition \ref{Def1b} reduces to the standard superposition rule notion.
\end{Note}

\begin{definition} We say that a system of ordinary differential equations (\ref{HODE}) of the order  $s$
is a {\it HODE Lie system} if the first-order system
\begin{equation}\label{FirstHOrd}
\left\{\begin{aligned}\frac{dx^i}{dt}&=y^{(1)i},\cr
\frac{dy^{(1)i}}{dt}&=y^{(2)i},\cr
\ldots&=\ldots,\cr
\frac{dy^{(s-1)i}}{dt}&=F^i\left(t,x,y^{(1)},\ldots,y^{(s-1)}\right),\end{aligned}\right.\qquad i=1,\ldots,n,
\end{equation}
obtained from (\ref{HODE}) by adding the new variables $y^{(j)i}= d^jx^i/dt^j$, with $i=1,\ldots,n$ and
$j=1,\ldots,s-1$, is a Lie system.
\end{definition}

Observe that the results and definitions described in previous sections can be directly generalized to systems
of HODEs. This is why, instead of detailing such generalisations, we shall merely describe a simple but
relevant result ensuring the existence of superposition rules for HODE Lie systems. In next sections, this
result is used to determine a superposition rule for second- and third-order Kummer--Schwarz equations.

\begin{proposition} Every HODE Lie system admits a superposition rule.
\end{proposition}
\begin{proof} Note that every solution of system (\ref{FirstHOrd}) is of the form ${\bf t}^{s-1}x_p(t)$
for a particular solution $x_p(t)$ of system (\ref{HODE}) and vice versa. Consequently, the superposition rule
$\Phi:({\rm T}^{s-1}\mathbb{R}^n)^m\times\mathbb{R}^{sn}\rightarrow {\rm T}^{s-1}\mathbb{R}^{n}$ for (\ref{FirstHOrd})
allows us to write that the general solution ${\rm t}^{s-1}x(t)$ of (\ref{FirstHOrd}) in the form
\begin{equation*}
{\bf t}^{s-1}x(t)=\Phi({\bf t}^{s-1}x_{(1)}(t),\ldots,{\bf t}^{s-1}x_{(m)}(t);k_1,\ldots,k_{sn}).
\end{equation*}
in terms of generic families of particular solutions $x_{(1)}(t),\ldots,x_{(m)}(t)$ of (\ref{HODE}), their
derivatives up to $s-1$ order, and the constants $k_1,\ldots,k_{sn}$. Applying the projection $\pi:{\rm
T}^{s-1}\mathbb{R}^n\rightarrow\mathbb{R}^n$ to both sides of the above relation, it follows that the general
solution $x(t)$ of system (\ref{HODE}) can be written as
\begin{equation*}
x(t)=(\pi\circ \Phi)({\bf t}^{s-1}x_{(1)}(t),\ldots,{\bf t}^{s-1}x_{(m)}(t);k_1,\ldots,k_{sn}).
\end{equation*}
In other words, $\Upsilon=\pi\circ\Phi:({\rm
T}^{s-1}\mathbb{R}^n)^m\times\mathbb{R}^{sn}\rightarrow\mathbb{R}^n$ is a superposition rule for
(\ref{HODE}).
\end{proof}

\section{Examples of superposition rules for systems of SODEs}\label{Examples}
\setcounter{equation}{0}
Let us illustrate now the results described in the previous sections by means of various examples extracted from
the physics and mathematics literature. As the first simple instance, consider the $n$-dimensional isotropic
harmonic oscillator
\begin{equation}\label{HO}
\frac{d^2x^{i}}{dt^2}=-\omega^2(t)x^i,\qquad i=1,\ldots,n,
\end{equation}
with a time-dependent frequency $\omega(t)$. This system appears broadly in the physics literature. For
instance, it occurs in the study of the fluctuations of the tachyon field obtained by using effective
Lagrangians \cite{ACT04,Kl04}, in the description of the movement of a particle on a heated spring
\cite{JS98}, in the analysis of the properties of diverse interesting nonlinear differential equations, like
Milne--Pinney equations, with many applications in physics \cite{CL08Diss,AL08,HG99}.

As shown in Section \ref{Justify}, systems (\ref{HO}) admit a base-superposition rule. Consequently,
Proposition \ref{Free} ensures that the systems of the form (\ref{HO}) must be SODE Lie systems. Actually,
this can be proved easily. In view of Definition \ref{DefiSODE}, demonstrating that each system of the form
(\ref{HO}) is a SODE Lie system reduces to proving that every first-order system
\begin{equation}\label{HarSysFirst}
\left\{\begin{aligned}\frac{dx^i}{dt}&=v^i,\\
\frac{dv^i}{dt}&=-\omega^2(t)x^i,\end{aligned}\right.\qquad i=1,\ldots,n,
\end{equation}
is a Lie system. Any system of the above type describes integral curves of the time-dependent vector field
$X_t=X_1+\omega^2(t)X_3$, with
\begin{equation*}
X_1=\sum_{i=1}^nv^i\frac{\partial}{\partial x^i},\quad
X_2=\frac 12\sum_{i=1}^n\left(x^i\frac{\partial}{\partial x^i}-v^i\frac{\partial}{\partial v^i}\right),\quad X_3=-\sum_{i=1}^nx^i\frac{\partial}{\partial v^i},
\end{equation*}
spanning a Lie algebra $V_{HO}$ of vector fields isomorphic to $\mathfrak{sl}(2,\mathbb{R})$ \cite{SIGMA}. It turns out that every system (\ref{HarSysFirst}) is a Lie system
and, therefore, the isotropic harmonic oscillators with time-dependent frequency (\ref{HO}) are SODE Lie
systems, as Proposition \ref{QuasiSup} states.

As equations (\ref{HO}) admit the same base-superposition rule, Proposition \ref{IncField}
ensures the existence of a finite-dimensional Lie algebra including $V_{HO}$ and $\Delta_{L}=\sum_{i=1}^nv^i\partial/\partial v^i$. Indeed, it is a straightforward computation to check that $\Delta_L$, $X_1$, $X_2$, and $X_3$ generate a Lie
algebra of vector fields isomorphic to $\mathfrak{gl}(2,\mathbb{R})$.

Let us now turn to exemplifying that superposition rules for systems of SODEs need not be invariant under
time-reparametrisations, as pointed out in  Section \ref{TheorySODE}. Recall that Milne--Pinney equations
(\ref{MilnePinney}) admit a quasi-base superposition rule depending on a constant of motion $I_3$. A
time-reparametrisation $\tau=\int^t_0e^{F(t')}dt'$ transforms (\ref{MilnePinney}) into the dissipative Milne--Pinney equation \cite{Sr86}-\cite{ABCCN97}
\begin{equation}\label{DMP}
\frac{d^2x}{d\tau^2}=-\frac{dF}{dt}(t(\tau))e^{-F(t(\tau))}\frac{dx}{d\tau}-\omega^2(t(\tau))e^{-2F(t(\tau))}x+\frac{ce^{-2F(t(\tau))}}{x^3}\,,
\end{equation}
and the superposition rule (\ref{MilnSup}) for Milne--Pinney equations yields that the general solution of
the above equation can be cast in the form (\ref{MilnSup}) again,  but in terms of a new constant $I_3$ reading
$$I_3=e^{2F(t)}\left(\frac{dx_{(1)}}{d\tau}(t)x_{(2)}(t)-\frac{dx_{(2)}}{d\tau}(t)x_{(1)}(t)\right)^2+
c\left[\left(\frac{x_{(1)}(t)}{x_{(2)}(t)}\right)^2+\left(\frac{x_{(2)}(t)}{x_{(1)}(t)}\right)^2\right].$$
This proves that the (quasi-base) superposition rule for Milne--Pinney equations is not invariant under
time-reparametrisations. Moreover, as SODEs (\ref{DMP}) admit, generically, a time-dependent superposition
rule (see \cite{CL10SecOrd2} for details), they need not be SODE Lie systems. Indeed, it was proved in
\cite{CL08Diss} that systems (\ref{DMP}) are not SODE Lie systems. Then, from Lemma \ref{CondNec},
Milne--Pinney equations cannot admit a base-superposition rule.

Let us now derive a new superposition rule for a relevant type of (nonautonomous) second-order differential
equation: the second-order Kummer--Schwarz equation
\begin{equation}\label{KS2}
\frac{d^2x}{dt^2}=\frac 3{2x} \left(\frac{dx}{dt}\right)^2-2b_0x^3+2a_0(t)x,
\end{equation}
where $b_0$ is a constant and $a_0(t)$ is an arbitrary time-dependent function. The study of these, hereafter
KS-2 equations, is motivated by its appearance in the theory of superposition rules \cite{Be07} and in
the analysis of second-order differential equations, where they appear related to the so-called {\it Kummer
problem} \cite{BR97}. These equations also appear associated with  the so-called second-order Gambier equation \cite{GGG11} and can be used to describe certain cosmological problems \cite{NR02}. Moreover, the solution of several cases of KS-2 equations amounts us to solving certain Milne--Pinney and Riccati equations \cite{Co94,GGG11,Nonlinear}. As these equations are ubiquitous in the physical literature, e.g. they appear in cosmology, quantum mechanics, classical mechanics \cite{NR02,Dissertationes}, the study of KS-2 equations can be consider as a useful approach to the analysis of these equations and their respective related physical problems.

In order to describe a superposition rule for KS-2 equations, we shall first prove that these equations are
SODE Lie systems, which ensures, by Proposition \ref{SODEsSup}, that they admit a superposition rule and
indicates how to derive it.

Recall that demonstrating that KS-2 equations are SODE Lie systems relies on proving that the first-order
system
\begin{equation}\label{FirstOrderKS2}
\left\{\begin{aligned}\frac{dx}{dt}&=v,\\
\frac{dv}{dt}&=\frac 32 \frac{v^2}x-2b_0x^3+2a_0(t) x,\end{aligned}\right.
\end{equation}
is a Lie system. To do this, consider the vector fields
\begin{eqnarray}\label{VFKS2}
 X_1=2x\frac{\partial}{\partial v},\quad X_2=x\frac{\partial}{\partial x}+2v\frac{\partial}{\partial v},\quad
X_3=v\frac{\partial}{\partial x}+\left(\frac 32\frac{v^2}x-2b_0x^3\right)\frac{\partial }{\partial v}.
\end{eqnarray}
Since
\begin{equation*}
[X_1,X_2]=X_1,\qquad [X_1,X_3]=2X_2,\qquad [X_2,X_3]=X_3,
\end{equation*}
they span a Lie algebra of vector fields isomorphic to $\mathfrak{sl}(2,\mathbb{R})$ and, as system (\ref{FirstOrderKS2})
is determined by the time-dependent vector field
$$
X_t=v\frac{\partial}{\partial x}+\left(\frac 32 \frac{v^2}x-2b_0x^3+2a_0(t)x\right)\frac{\partial}{\partial
v}=X_3+a_0(t)X_1,
$$
the second-order Kummer--Schwarz equations are SODE Lie systems.

It is interesting that, like time-dependent frequency harmonic oscillators, Kummer--Schwarz
equations are SODE Lie systems related to a Vessiot--Guldberg Lie algebra isomorphic to
$\mathfrak{sl}(2,\mathbb{R})$. This can be used to establish interesting relations between these equations and
other (SODE) Lie systems associated with the same Lie algebra (cf. \cite{SIGMA}).

Once it has been proved that KS-2 equations are SODE Lie systems, the following step toward deriving their
superposition rule is, in view of Lemma \ref{SODEsSup}, to determine the part of the standard superposition
rule for system (\ref{FirstOrderKS2}) describing the $x$ coordinate of its general solution. To do this, let us
apply the method described in Section \ref{FLS}.

The vector fields $X_1,X_2,X_3$ form a basis for a Vessiot--Guldberg Lie algebra of (\ref{FirstOrderKS2}) and
their prolongations to $({\rm T}\mathbb{R})^2$ are linearly independent at a generic point. Let $\widetilde
X_1$, $\widetilde X_2$, and $\widetilde X_3$ be diagonal prolongations to $({\rm T}\mathbb{R})^3$. As
$[\widetilde X_1,\widetilde X_3]=\widetilde{[X_1,X_3]}=2\widetilde X_2$, if a function $F:\left({\rm
T}\mathbb{R}\right)^3\rightarrow\mathbb{R}$ satisfies $\widetilde X_1F=\widetilde X_3F=0$, then $\widetilde
X_2F=0$. Thus, obtaining a common first-integral for $\widetilde X_1,\widetilde X_2,\widetilde X_3$ reduces
to finding a common first-integral for $\widetilde X_1$ and $\widetilde X_3$.

Consider the canonical coordinates $\{x_0,v_0,x_1,v_1,x_2,v_2\}$ in $({\rm T}\mathbb{R})^3$ and
suppose that the common first-integral $F$ for $\widetilde X_1$ and $\widetilde X_3$ depends only on the
variables $x_0,x_1,v_0$, and $v_1$. As $\widetilde X_1F=0$, the method of characteristics yields that $F$ must
be constant along the solutions of the characteristic system
\begin{equation}
\frac{dv_0}{x_0}=\frac{dv_1}{x_1},\qquad dx_0=dx_1=0.
\end{equation}
Integrating the above system, we find that that there exists a certain function
$F_2:\mathbb{R}^3\rightarrow\mathbb{R}$ such that
$F\left(x_0,x_1,v_0,v_1\right)=F_2\left(x_0,x_1,\xi=x_1v_0-x_0v_1\right)$. In terms of the variables
$x_0,x_1,\xi,v_1$, the condition $\widetilde X_3F=\widetilde X_3F_2=0$ reads
$$
\left(\frac{\xi+v_1x_0}{x_1}\right)\frac{\partial F_2}{\partial x_0}+v_1\frac{\partial F_2}{\partial
x_1}+\left[\frac 32\left(\frac{\xi^2+2\xi v_1x_0}{x_1x_0}\right)+2b_0(x_1^3x_0-x_0^3x_1)\right]\frac{\partial
F_2}{\partial \xi}=0.
$$
As $F_2$ does not depend on $v_1$, the above equation implies
$$
\frac{\xi}{x_1}\frac{\partial F_2}{\partial x_0}+\left[\frac{3
\xi^2}{2x_1x_0}+2b_0(x_1^3x_0-x_0^3x_1)\right]\frac{\partial F_2}{\partial \xi}=0, \quad
\frac{x_0}{x_1}\frac{\partial F_2}{\partial x_0}+\frac{\partial F_2}{\partial
x_1}+\frac{3\xi}{x_1}\frac{\partial F_2}{\partial \xi}=0.
$$
Applying again the method of characteristics to the second equation, we see that there exists a function
$F_3:\mathbb{R}^2\rightarrow\mathbb{R}$ such that $F_2(x_0,x_1,\xi)=F_3(K_1=x_0/x_1,K_2=x_1^3/\xi)$. Let us
express the first of the above equations using the variables $K_1, K_2$, and $\xi$. As a result, it  turns out that
$$
\widetilde X_3F_3=\frac{\xi}{x_1^2}\left(\frac{\partial F_3}{\partial K_1}-\left[\frac
3{2K_1}+2b_0K_2^2(K_1-K_1^3)\right]K_2 \frac{\partial F_3}{\partial K_2}\right)=0.
$$
The characteristic system for the preceding equation is
\begin{equation}\label{FirstInt}
dK_1=-\frac{dK_2}{K_2\left[\frac{3}{2K_1}+2b_0K_2^2(K_1-K_1^3)\right]},
\end{equation}
whose solution is
\begin{equation*}
\Gamma_1=\frac{(v_0x_1-v_1x_0)^2}{x_0^3x_1^3}+4b_0\frac{x_0^2+x_1^2}{x_0x_1}.
\end{equation*}
Hence, $\Gamma_1$ is a first-integral common to $\widetilde X_1$ and $\widetilde X_3$. Similarly, if we
suppose that $F$ depends only on  $x_0,x_2,v_0$, and $v_2$, or, alternatively, on $x_1,x_2,v_1$, and $v_2$, two
new first-integrals appear:
\begin{equation}\label{SecondInt}
\Gamma_2=\frac{(v_0x_2-v_2x_0)^2}{x_0^3x_2^3}+4b_0\frac{x_0^2+x_2^2}{x_0x_2},\qquad
\Gamma_3=\frac{(v_1x_2-v_2x_1)^2}{x_1^3x_2^3}+4b_0\frac{x_1^2+x_2^2}{x_1x_2}.
\end{equation}

Since $\partial(\Gamma_1,\Gamma_2)/\partial(x_0,v_0)\neq 0$ at a generic point of $({\rm T}\mathbb{R})^3$, the procedure described in Section \ref{FLS} allows us to determine the values of $x_{0}$ and $v_{0}$ in terms
of $x_{1}, x_{2}, v_1,v_2$, and two constants giving rise to a superposition rule for system
(\ref{FirstOrderKS2}). Indeed, fixing $\Gamma_1=k_1$ enables us to determine the value of $v_0$ in terms of
$x_0,x_1,v_1,$ and $k_1$. Substituting this value into the equation $k_2=\Gamma_2$ and with the aid of
$\Gamma_3$, we can express $x_0$ in terms of $x_1,x_2,k_1,k_2$, and $\Gamma_3$ as
\begin{equation}\label{QuasiSup}
 x_0=\frac{(\Gamma_3k_1-8b_0k_2)x_1+\!(\Gamma_3k_2-8b_0k_1)x_2\pm
2\lambda_{k_1,k_2}(\Gamma_3)[\Gamma_3x_1x_2-4b_0(x_1^2+x_2^2)]^{1/2}}{16 b_0 \Gamma_3
+x^{-1}_1x^{-1}_2[(k_1x_1-k_2x_2)^2-64b_0^2(x_1^2+x_2^2)]},
\end{equation}
where
\begin{equation*}
\lambda_{k_1,k_2}(\Gamma_3)=\left[256b_0^3+k_1k_2\Gamma_3-4b_0(k_1^2+k_2^2+\Gamma_3^3)\right]^{1/2}.
\end{equation*}
Expressions (\ref{SecondInt}) ensure that $[\Gamma_3 x_1x_2-4b_0(x_1^2+x_2^2)]^{1/2}$ is real. Meanwhile, for
each pair of solutions $x_1(t),x_2(t)$, $k_1$ and $k_2$ must be chosen so that $\lambda_{k_1,k_2}(\Gamma_3)$ is
real.

Since $\Gamma_3$ depends on the variables $x_1,x_2,v_1,v_2,$  it is clear that expression (\ref{QuasiSup})
constitutes a part of a superposition rule for any system of the form (\ref{FirstOrderKS2}), describing the
component $x$ of its general solution in terms of two particular solutions $x_1(t)$, $x_2(t)$ of (\ref{KS2}),
their derivatives $v_1(t),v_2(t)$, and two constants $k_1$ and $k_2$. Therefore, in view of Lemma \ref{Free},
this allows us to write the general solution $x(t)$ of the equation (\ref{KS2}) in terms of two particular
solutions, their derivatives, and two constants. This provides us with a superposition rule $\Upsilon:({\rm\bf t}^1_{x_1},{\rm \bf t}^1_{x_2};k_1,k_2)\in({\rm
T}\mathbb{R})^2\times\mathbb{R}^2\mapsto x=\Upsilon({\rm\bf t}^1_{x_1},{\rm \bf t}^1_{x_2};k_1,k_2))\in\mathbb{R}$ for KS-2 equations of the form
$$
x=\frac{(Ik_1-8b_0k_2)x_1+(Ik_2-8b_0k_1)x_2-2\lambda_{k_1,k_2}(I)[Ix_1x_2-4b_0(x_1^2+x_2^2)]^{1/2}}{16
b_0 I +x^{-1}_1x^{-1}_2[(k_1x_1-k_2x_2)^2-64b_0^2(x_1^2+x_2^2)]},
$$
where $I=\Gamma_3$ is regarded as a function of the variables of $({\rm T}\mathbb{R})^2$. Note in addition
that the above expression can also be naturally considered as a quasi-base superposition rule of the form
$G(x_1,x_2,I;k_1,k_2)$ for KS-2 equations.
\section{A superposition rule for the third-order Kummer--Schwarz equations}
\setcounter{equation}{0}
The present section is devoted to the study of third-order Kummer--Schwarz equations \cite{Be07,Be88,Be82} of the
form
\begin{equation}\label{KS3}
\frac{d^3x}{dt^3}=\frac 32
\left(\frac{dx}{dt}\right)^{-1}\left(\frac{d^2x}{dt^2}\right)^2-2b_0\left(\frac{dx}{dt}\right)^3+2
a_0(t)\frac{dx}{dt},
\end{equation}
with $b_0$ being a constant and $a_0(t)$ being any time-dependent function. Our aim is to exemplify how the results
of Section \ref{HODEs} can be applied to investigate a relevant third-order differential equation. As a
result, it is shown that third-order Kummer--Schwarz equations, hereby KS-3, are HODE Lie systems, and an
interesting relation to KS-2, Riccati, and Milne--Pinney equations is pointed out. Finally, a new superposition
rule for KS-3 equations depending on a single particular solution and three constants is derived.

The relevance of the study of KS-3 equations relies, for instance, on their relation to the so-called {\it
Kummer's problem} \cite{Be07,Be88,Be82}, Milne--Pinney equations \cite{AL08}, and Riccati equations
\cite{AL08,Co94,EEL07}. These relations can be used to study multiple physical systems described by these latter equations through KS-3 equations, e.g. the case of quantum non-equilibrium dynamics of many-body
systems \cite{BDG09}. Furthermore, Kummer--Schwarz equations with $b_0=0$ can be rewritten as $\{x,t\}=2a_0(t)$, where $\{x,t\}$ is the so-called {\it Schwarzian derivative} of the function $x(t)$ with respect to $t$ \cite{LG99}.

In order to study KS-3 equations, let us define  $y^{(1)}=dx/dt$, $y^{(2)}=d^2x/dt^2$, and write equation
(\ref{KS3}) in the form
\begin{equation}\left\{\label{FirstOrderKummer}
\begin{aligned}\frac{dx}{dt}&=y^{(1)},\\
\frac{dy^{(1)}}{dt}&=y^{(2)},\\
\frac{dy^{(2)}}{dt}&=\frac 32 \frac{y^{(2)2}}{y^{(1)}}-2b_0y^{(1)3}+2a_0(t)y^{(1)}.\end{aligned}\right.
\end{equation}
Consider the vector fields $X_1$, $X_2$, and $X_3$ on ${\rm T}^2\mathbb{R}$,
\begin{equation}\label{VFKS1}
\begin{aligned}X_1=2y^{(1)}\frac{\partial}{\partial y^{(2)}},\qquad
X_2=y^{(1)}\frac{\partial}{\partial y^{(1)}}+2y^{(2)}\frac{\partial}{\partial y^{(2)}},\\
X_3=y^{(1)}\frac{\partial}{\partial x}+y^{(2)}\frac{\partial}{\partial y^{(1)}}+\left(\frac 32
\frac{y^{(2)2}}{y^{(1)}}-2b_0y^{(1)3}\right)\frac{\partial}{\partial y^{(2)}},\end{aligned}
\end{equation}
satisfying the commutation relations
\begin{equation*}
[X_1,X_3]=2X_2,\qquad [X_2,X_3]=X_3,\qquad [X_1,X_2]=X_1.
\end{equation*}
Obviously, these vector fields span a Lie algebra of vector fields isomorphic to
$\mathfrak{sl}(2,\mathbb{R})$. Additionally, system (\ref{FirstOrderKummer}) describes integral curves of the
time-dependent vector field $X_t=X_3+a_0(t)X_1$. Thus, KS-3 equations are HODE Lie systems. Let us derive a
superposition rule for them.

The vector fields $X_1, X_2, X_3$ are linearly independent at a generic point of ${\rm T}^2\mathbb{R}$.
Therefore, the diagonal prolongations $\widetilde X_1, \widetilde X_2, \widetilde X_3$ of $X_1,X_2,X_3$ to $({\rm
T}^2\mathbb{R})^2$ span a generalized distribution $\mathcal{D}$. Such a generalized
distribution is three-dimensional in a neighbourhood of a generic point, where the distribution becomes
regular. Hence, the vector fields of the distribution admit, at least locally, three common first-integrals.
As $[X_1,X_3]=2X_2$, we have $[\widetilde X_1,\widetilde X_3]=2\widetilde X_2$, and obtaining first-integrals
common for all the vector fields of $\mathcal{D}$ reduces to determining first-integrals common for
$\widetilde X_1$ and $\widetilde X_3$.

Let us first analyse first-integrals of the vector field $\widetilde X_1$ on $({\rm T}^2\mathbb{R})^2$, i.e.
solutions $F:\left({\rm T}^2\mathbb{R}\right)^2\rightarrow\mathbb{R}$ of the equation
\begin{equation*}
\widetilde X_1F=2y_0^{(1)}\frac{\partial F}{\partial y_0^{(2)}}+2y_1^{(1)}\frac{\partial F}{\partial
y_1^{(2)}}=0.
\end{equation*}
The method of characteristics shows that the first-integrals of the above vector field are
functions constant along the solutions of the so-called characteristic system of $\widetilde X_1$, namely,
\begin{equation*}
\frac{dy^{(2)}_0}{y^{(1)}_0}=\frac{dy^{(2)}_1}{y^{(1)}_1},\qquad dy^{(1)}_0=dy^{(1)}_1=dx_0=dx_1=0.
\end{equation*}
Such solutions are given, in an implicit form, by the algebraic equations $\xi=y^{(1)}_0y^{(2)}_1-y^{(1)}_1y^{(2)}_0,\,\,v_0=y^{(1)}_0,\,\,v_1=y^{(1)}_1,\,\,y_0=x_0,\,\,y_1=x_1$, where $\xi,v_0,v_1,y_0,y_1$ are certain real constants. In other words, any first-integral $F:\left({\rm
T}^2\mathbb{R}\right)^2\rightarrow\mathbb{R}$ of the vector field $\widetilde X_1$ depends only on the
previous variables. Hence, there exists a function $F_2:\mathbb{R}^5\rightarrow\mathbb{R}$ such that
$F(x_0,x_1,y_0^{(1)},y^{(1)}_1,y_0^{(2)},y^{(2)}_1)=F_2(y_0,y_1,v_0,v_1,\xi)$.

Remember that we are interested in determining a common first-integral for the vector fields $\widetilde X_1$
and $\widetilde X_3$. In view of the above result, $\widetilde X_3F_2=0$ amounts us to
\begin{equation*}
\begin{aligned} \sum_{a=0,1}v_a\frac{\partial F_2}{\partial y_a}\!+\left(\frac{v_0a_1-\xi}{v_1}\right)\frac{\partial F_2}{\partial v_0}+\!a_1\frac{\partial F_2}{\partial
v_1}+\!\left(\frac{3\xi a_1}{v_1}-\frac{3\xi^2}{2v_1v_0}+2 b_0(v_0^3v_1-v_1^3v_0)\right)\frac{\partial
F_2}{\partial \xi}=0,\end{aligned}
\end{equation*}
where we defined $a_1\equiv y_1^{(2)}$. In terms of the vector fields
$$
\begin{aligned}
Z_1=v_0\frac{\partial}{\partial y_0}+v_1\frac{\partial}{\partial y_1}-\frac{\xi}{v_1}\frac{\partial}{\partial v_0}+
\left(-\frac{3\xi^2}{2v_1v_0}+2 b_0(v_0^3v_1-v_1^3v_0)\right)\frac{\partial}{\partial \xi},\\
Z_2=v_0\frac{\partial}{\partial v_0}+v_1\frac{\partial}{\partial v_1}+{3\xi}\frac{\partial}{\partial \xi},
\end{aligned}
$$
we can write $\widetilde X_3F_2=Z_1F_2+\frac{a_1}{v_1}Z_2F_2=0$. Since $F_2$ does not depend on $a_1$, the previous decomposition implies  $Z_1F_2=Z_2F_2=0$. Applying the method of characteristics to
$Z_2F_2=0$, we find that $F_2$ must be constant along solutions of the characteristic system
\begin{equation*}
\frac{dv_0}{v_0}=\frac{dv_1}{v_1}=\frac{d\xi}{3\xi},\qquad dy_0=dy_1=0.
\end{equation*}
Therefore, $F_2$ depends only on the variables $K_1=v_1/v_0$, $K_2=v_1^3/\xi$, $y_0,y_1$, i.e.
there exists a function $F_3:\mathbb{R}^4\rightarrow\mathbb{R}$ such that
$F(x_0,x_1,v_0,v_1,a_0,a_1)=F_3(y_0,y_1,K_1,K_2)$.

To obtain a common first-integral for all the vector fields in $\mathcal{D}$, it remains to impose
$Z_1F=Z_1F_3=0$. In the coordinate system $\{y_0,y_1,K_1,K_2,\xi,a_1\}$, this equation reads
\begin{equation*}
 \qquad \xi^{1/3}K_2^{1/3}\left(\frac{1}{K_1}\frac{\partial F_3}{\partial y_0}+\frac{\partial F_3}{\partial
y_1}+\frac{K_1^2}{K_2}\frac{\partial F_3}{\partial
K_1}+\right.\left.\left[\frac{3K_1}{2}-2b_0K_2^2\left(K_1^{-3}-K_1^{-1}\right)\right]\frac{\partial
F_3}{\partial K_2}\right)=0.
\end{equation*}
The characteristic system corresponding to the above equation is
\begin{equation*}
dy_0=\frac{dy_1}{K_1}=\frac{K_2dK_1}{K_1^3}=\frac{K_1^2dK_2}{\frac{3}{2} K_1^4-2b_0K^2_2(1-K_1^{2})}.
\end{equation*}
From the last equality of the above system, we obtain
\begin{equation*}
\frac{dK_2}{dK_1}=\frac{3K_2}{2K_1}-\frac{2b_0}{K_1^5}(1-K_1^2)K_2^3\longrightarrow K_2(K_1)=\pm
\frac{K_1^2}{\sqrt{K_1\Gamma_1-4 b_0(1+K_1^2)}},
\end{equation*}
for a certain real constant $\Gamma_1$. Consequently, a common first-integral of the vector fields of
$\mathcal{D}$ reads
\begin{equation*}
\Gamma_1=\frac{K_1^4+4b_0K_2^2(1+K_1^2)}{K_1K_2^2}.
\end{equation*}
Now, $dy_1=K_2dK_1/K_1^2$ and the above expression yields
\begin{equation}\label{Gamma2}
\frac{dy_1}{dK_1}={\rm sg}(K_2){(\Gamma_1 K_1-4b_0(1+K_1^2))}^{-1/2},
\end{equation}
where ${\rm sg}$ stands for the {\it sign function}. Assume, for simplicity, ${\rm sg}(K_2)=1$ and
$b_0<0$. Hence,
\begin{equation*}
\Gamma_2=\left(-8b_0K_1+\Gamma_1+4\sqrt{4b^2_0(1+K_1^2)-b_0K_1\Gamma_1}\right)e^{-2x_1\sqrt{-b_0}}
\end{equation*}
is a second first-integral. Likewise, from $\Gamma_1$ and expression
$dy_0=K_2dK_1/K_1^3$, one gets
\begin{equation*}
\frac{dy_0}{dK_1}={\left(K_1\sqrt{\Gamma_1 K_1-4b_0(1+K_1^2)}\right)}^{-1},
\end{equation*}
i.e.
\begin{equation*}
\Gamma_3=y_0-\frac{1}{2\sqrt{-b_0}}\ln\left[\frac{2\sqrt{-b_0}K_1}{-8b_0+K_1\Gamma_1+4\sqrt{4b_0^2(1+K_1^2)-b_0K_1\Gamma_1}}\right],
\end{equation*}
is another first-integral. As $\partial(\Gamma_1,\Gamma_2,\Gamma_3)/\partial(x_0,y^{(1)}_0,y^{(2)}_0)\neq 0$ at a generic point of $({\rm T}^2\mathbb{R})^2$, fixing $k_1=\Gamma_1$, and $k_2=\Gamma_2$, we can easily express $K_1$ in terms of
$k_1, k_2$ and $x_1$. Using this and putting $k_3=\Gamma_3$, we obtain
$$
x_0=k_3+\ln\left[\frac{2\sqrt{-b_0}[64 b_0^2-f^2_{k_1,k_2}(x_1)]}{64
b_0^2(k_1-2e^{2\sqrt{-b_0}x_1}k_2)-k_1f^2_{k_1,k_2}(x_1)+8b_0(64b_0^2-k_1^2+e^{4\sqrt{-b_0}x_1}k_2^2)}\right]^{\frac{1}{2\sqrt{-b_0}}},
$$
where $f_{k_1,k_2}(x_1)=k_1-e^{2\sqrt{-b_0}x_1}k_2$. Note that the above expression is a superposition rule for
third-order Kummer--Schwarz equations, which provides the general solution $x_0(t)$ of any instance of such
equations in terms of a generic particular solution $x_1(t)$ and the constants $k_1,k_2,k_3$. Obviously, this
represents an improvement with respect to other similar expressions for KS-3 equations, which allows us to
describe their general solutions in terms of two particular solutions of a time-dependent frequency harmonic
oscillator \cite{Be07}. In addition, this expression is an instance of a quasi-base superposition rule for a
third-order differential equation.
\section{Conclusions and Outlook}
\setcounter{equation}{0}
We have proposed and analysed a general concept of a superposition rule for systems of HODEs. Some specific
types of such superposition rules that appear in the literature have been studied and other new types have been
introduced and investigated. All our results have been illustrated with examples extracted from the mathematical
and physics literature. In particular, two new superposition rules for second- and third-order Kummer--Schwarz
equations have been derived.

There are still many open questions concerning the properties of
superposition rules for systems of HODEs. For instance, it would be interesting to find methods for analysing
the existence of solutions of system (\ref{Char}), which would facilitate the determination of the existence
of superposition rules for systems of SODEs. Additionally, it would be interesting to
apply the methods developed here to analyse first-order systems from a new perspective.

In the future, we intend to study the whole {\it Riccati hierarchy} \cite{GL99}, some of whose members, like second-order
Riccati equations, have already been analysed by means of the theory of Lie systems \cite{CL10SecOrd2}. Further, we aim to apply our results in the analysis of soliton solutions of PDEs described by the Riccati hierarchy \cite{GL99}. Additionally, we plan to employ the theory of Lie systems so as to geometrically explain the relation of Kummer--Schwarz equations to Lie systems associated with a Vessiot--Guldberg Lie algebra isomorphic to $\mathfrak{sl}(2,\mathbb{R})$. This may be used to clarify their known
connections with time-dependent harmonic oscillators or Riccati equations \cite{Co94,Be88} as well as to establish new ones. These and other topics
will be analysed in forthcoming works.

\section{Acknowledgements}

Partial financial support by research projects MTM2009-11154, MTM2010-12116-E, and E24/1 (DGA) are
acknowledged. Research of the second author financed by the Polish Ministry of Science and Higher Education
under the grant N N201 365636. J. de Lucas also acknowledges financial support by DGA under project  FMI43/10
to accomplish a research stay in the University of Zaragoza.

\end{document}